\documentclass[onefignum,onetabnum]{siamonline171218}

\usepackage{amssymb,amsmath,mathtools,amsbsy}

\usepackage{graphicx,rotate}
\usepackage{tikz}
\usepackage{subcaption}

\usepackage{float}
\usepackage{enumerate}
\usepackage{dsfont}

\usepackage{accents}

\usepackage{bigstrut}
\usepackage{multirow}
\usepackage{booktabs}

\usepackage{tikz}

\usepackage[normalem]{ulem}

\renewcommand{\theequation}{\arabic{section}.\arabic{equation}~}

\numberwithin{equation}{section}

\theoremstyle{definition}

\newsiamthm{remark}{\bf Remark}

\newcommand{\E}{\mathbb{E}}

\newcommand{\bm}{\pmb}

\renewcommand{\P}{\mathbb{P}}

\newcommand{\T}{\intercal}

\newcommand{\Ahat}{\widehat{A}}
\newcommand{\Mhat}{\widehat{M}}

\newcommand{\D}{\mathcal{D}}
\newcommand{\A}{\mathcal{A}}

\renewcommand{\hbar}{\overline{h}}
\newcommand{\nubar}{\overline{\nu}}
\newcommand{\qbar}{\overline{q}}

\newcommand{\nut}{{\tilde{\nu}}}
\newcommand{\mfK}{{\mathcal{K}}}
\newcommand{\Pk}{{\P^k}}

\newcommand{\g}{\bm{g}}

\newcommand{\qj}[1]{q^{j,{#1}}}
\newcommand{\qtk}[1]{\tilde{q}^{k,{#1}}}

\def\F{{\mathcal{F}}}
\def\G{{\mathcal{G}}}

\def\N{(N)}
\def\mfQ{{\mathfrak{Q}}}
\def\mfN{{\mathfrak{N}}}
\def\P{{\mathbb{P}}}

\def\R{\mathbb{R}}
\def\nuj{\nu^j}
\def\nujst{\nu^{j,\ast}}
\def\nubarN{\nubar^{\N}}
\def\nubarkN{\nubar^{k,\N}}
\def\bnubar{\bm{\nubar}}
\def\qbar{\bar{q}}
\def\bqbar{\bm{\qbar}}
\def\qt{\tilde{q}}
\def\bqt{\bm{\qt}}

\def\ba{\bm{a}}
\def\bLambda{\bm{\Lambda}}
\def\bPsi{\bm{\Psi}}
\def\bphi{\bm{\phi}}
\def\bg{\bm{g}}
\def\blambda{\bm{\lambda}}

\def\nubark{\nubar^k}
\def\Hbar{\overline{H}}
\def\nutk{\nut^k}
\def\mM{\mathcal{M}}
\def\mMbar{\overline{\mM}}
\def\bmMbar{\bm{\overline{\mM}}}
\def\bmMt{\bm{\tilde{\mM}}}
\def\bnut{\bm{\nut}}
\def\bgt{\tilde{\bg}}
\def\bxi{\bm{\xi}}
\def\HT{\mathbb{H}^2_T}
\def\invmean{\overline{m}}
\def\LT{\mathbb{L}^2_T}

\newcommand{\smallqfm}[1]{\left( \begin{smallmatrix} {#1}_t \\ \qj{#1}_t \end{smallmatrix} \right)}

\definecolor{WowColor}{rgb}{.75,0,.75}
\definecolor{SubtleColor}{rgb}{0,0,.50}

\newcounter{margincounter}

\headers{MFG with Partial Information for Algorithmic Trading}{Casgrain, P. and Jaimungal, S.}

\title {
Mean Field Games with Partial Information for Algorithmic Trading\thanks{SJ would like to acknowledge the support of the Natural Sciences and Engineering Research Council of Canada (NSERC), [funding reference numbers RGPIN-2018-05705 and RGPAS-2018-522715]}
}

\author{
Philippe Casgrain\thanks{Department of Statistical Sciences, University of Toronto, Canada (\email{p.casgrain@utoronto.ca}).}
\and
Sebastian Jaimungal\thanks{Department of Statistical Sciences, University of Toronto, Canada (\email{sebastian.jaimungal@utoronto.ca}, \url{http://sebastian.statistics.utoronto.ca.})}
}

\begin{document}

\maketitle

\begin{abstract}
Financial markets are often driven by latent factors which traders cannot observe. Here, we address an algorithmic trading problem with collections of heterogeneous agents who aim to perform optimal execution or statistical arbitrage, where all agents filter the latent states of the world, and their trading actions have permanent and temporary price impact. This leads to a large stochastic game with heterogeneous agents. We solve the stochastic game by investigating its mean-field game (MFG) limit, with sub-populations of heterogeneous agents, and, using a convex analysis approach, we show that the solution is characterized by a vector-valued forward-backward stochastic differential equation (FBSDE). We demonstrate that the FBSDE admits a unique solution, obtain it in closed-form, and characterize the optimal behaviour of the agents in the MFG equilibrium. Moreover, we prove the MFG equilibrium provides an $\epsilon$-Nash equilibrium for the finite player game. We conclude by illustrating the behaviour of agents using the optimal MFG strategy through simulated examples.
\end{abstract}

\section{Introduction}

Financial markets are incredibly complex systems that have a significant impact on how our society functions. One main source of complexity is the continuous interaction of millions of traders (agents) all interacting simultaneously. Another is the effect of latent factors which drive prices and hence optimal decisions. To analyse such systems through the lens of a model, we must consider how interactions among agents affect decisions of each individual market participant. Furthermore, the population of market participants may be heterogeneous in their trading objectives and their behaviour types.   To perform optimally, each agent must formulate a model of asset prices which takes into account latent information sources, as well as the trading decisions of other market participants.

How agents should optimally process this large amount of information into trading decisions is a question that the algorithmic trading literature has long attempted to answer. Classically, algorithmic trading models only consider a single agent interacting with a stochastically evolving asset price, as in \cite{almgren2001optimal}. Later works such as \cite{cartea2016incorporating} study algorithmic trading where other (multiple) agents interact with the price, but not necessarily in an optimal manner, nor directly with the agent's actions.

To effectively model  the interaction of agents in electronic markets, we instead turn towards a mean-field game (MFG) approach, which, in general, aims at approximating the optimal actions of large populations of non-corporative interacting agents in game-like settings. A large body of research has already been devoted to the study of MFGs. The original works stem from \cite{huang2006large}, \cite{HuangCaines_TAC07}, and \cite{lasry2007mean}. Among the many extensions and generalizations which explore the broad theory of MFGs as well as their applications, we highlight the following works:~\cite{huang2010large}, \cite{gueant2011mean}, \cite{nourian2013e}, and \cite{carmona2013probabilistic}. This theory has seen application in various financial contexts, such as \cite{carmona2013mean} and \cite{huang2017robust} who use it to model systemic risk, \cite{jaimungal2015mean} show use it for algorithmic trading in the presence of a major agent and a population of minor agents,  \cite{cardaliaguet2016mean} who investigate MFG in the context of optimal execution, and \cite{firoozi2015varepsilon} who look at mean-field games in algorithmic trading with partial information on states. \cite{bayraktar2017mini} study a model for mini-flash crashes where agents interact, but each agent optimizes a problem that ignores the optimal actions of other agents, while realized prices incorporate the actions of all agents.

In contrast to other work on MFGs, as well as its specific application to algorithmic trading, here, motivated by \cite{casgrain_jaimungal_2016}, we include latent states so that agents do not have full information about the system dynamics. Under a very general specification for the asset price dynamics, and allowing for heterogeneity in behaviours and objectives, we provide an optimal strategy for all participating agents. This strategy is able to make effective use of market price path data to filter out information about latent processes, and simultaneously account for the actions of all other (heterogenous) market participants.

Our work sets itself apart from the current literature by approaching the algorithmic trading problem using a MFG approach, as well, we avoid the assumption that each agent observes the strategies of other agents, and due to the latent factor, agents do not have complete knowledge of the model driving asset returns. We also take a novel approach to solving the MFG by applying convex analysis techniques directly to the problem, rather than relying on the Stochastic Maximum Principle or the dynamic programming principle. Our approach yields powerful results such as a single optimality equation that directly characterizes the optimal control without the use of auxiliary processes.  Furthermore, we include a latent structure in our model, causing agents to have incomplete information on the dynamics of the market.

The remainder of the paper is structured as follows. Section~\ref{sec:Model} presents our stochastic game model for the finite-population market. Section~\ref{sec:solving-the-mfg} begins by formulating a MFG version of the infinite player limit. The section continues by applying convex analysis tools to characterize the optimal trading strategy as an infinite system of Forward-Backward Stochastic Differential Equations (FBSDEs) coupled through a collection of mean-field processes. The system of FBSDEs is then solved, where we provide closed-form solutions of both the individual agent's optimal control and for the mean-field processes. In Section~\ref{sec:epsilon-Nash}, we show that the optimal solution to the MFG satisfies an $\epsilon$-Nash equilibrium property for the finite population model. Lastly, Section~\ref{sec:Numerical-Experiments} explores some of the optimal strategy's behaviour by analyzing simulated games. Section~\ref{sec:Conclusion} concludes, and the appendices contain most proofs of the results presented in the body of the paper.

\section{Model Setup and Motivation} \label{sec:Model}

In this section, we present a stochastic game where a collection of agents all trade a single asset\footnote{It is possible to generalize to trading multiple assets which have fairly arbitrary dynamics. We, however, opt to restrict the analysis to the single asset case to ensure the key insights on how agents interact, and on how latent information is incorporated, are not obfuscated by the interaction of multiple assets.}. Agents interact with one another by affecting the price of the asset through their buying and selling decisions. The model assumes agents trade at a continuous rate over a specified time period, and mark-to-market the value of their position at the end of the trading horizon. In contrast to most other works, we allow for agent heterogeneity by defining sub-populations of players each with their own unique behavioural parameters.

Each agent seeks to maximize a functional which measures their performance over the course of the trading period. The flow of information available to each agent contains (i) the asset price, but not latent processes which drive the price, (ii) their own state, and (iii) their own actions (in particular we exclude the information about other agents inventories or trading strategies). This section concludes by presenting the agent's optimization problem and by formally describing the Nash equilibrium we seek.

\subsection{The Agent's State Processes}

Define the filtered probability space \\ \noindent $\left( \Omega, \G, \mathfrak{G} = \{ \G_t \}_{t\in[0,T]} , \P \right)$, where $T>0$ is a fixed and finite time horizon. All processes defined in this section are adapted to the filtration $\G_t$ unless otherwise stated. We assume that the population of traders is composed of $N\in\mathbb{N}$ individual agents, each indexed by an integer $j\in\mfN:=\{1,\dots,N\}$, and all agents are trading in a single asset. To allow for different trading behaviour within this large population, the population is divided into $K \leq N$ disjoint sub-populations indexed by $k\in\mfK:=\{1,\dots,K\}$, where all traders within a sub-population are assumed to behave homogeneously. Let us define
\begin{equation}
	\mfK_k^{\N} = \left\{ \, j\in \mfN \,:\,  \text{agent }j \text{ is in sub-population } k \, \right\}
\end{equation}
to represent the collection of agents that belong to sub-population $k$, where the superscript $\N$ is included to show explicit dependence on the total number of agents. We also define $N_k^{\N} = \# \mfK_k^{\N}$ to be the total number of agents within sub-population $k$, and assume the proportion of the total population contained in each sub-population remains finite in the limit as $N$ becomes large, so that
\begin{equation} \label{eq:limit-assumption-proportions}
	\lim_{N\nearrow\infty} \frac{N_k^{\N}}{N} = p_k \in (0,1)
	\, .
\end{equation}

Each agent $j$ controls the rate at which they buy or sell the traded asset via the process $(\nuj_t)_{t\in[0,T]}$, where $\nuj_t>0$ indicates buying and $\nuj_t<0$ indicates selling. The agents keep track of the net amount of shares they have accumulated via their (controlled) inventory process $\qj{\nuj} = ( \qj{\nuj}_t )_{t\in[0,T]}$ which satisfies
\begin{equation} \label{eq:agent-inventory-dynamics}
	\qj{\nuj}_t = \mfQ^j_0 + \int_0^t \nuj_u \, du 	\,,
\end{equation}
where $\{\mfQ^j_0\}_{j=1}^N$ is a collection of independent random variables representing agent-$j$'s inventory at the start of the trading period. We further assume the initial inventory positions of traders have a bounded variance, so that $\exists \; C>0$ (independent of $N$) such that $\E[ \left( \mfQ^j_0 \right)^2 ] < C$, $\forall\;j\in\mfN$. For a fixed strategy $\nuj$, the definition of $\qj{\nuj}_t$ in equation~\eqref{eq:agent-inventory-dynamics} corresponds to each agent $j$ buying or selling roughly $\epsilon \times \nuj_t$ units of the traded asset in any small time interval $[t,t+\epsilon]$. Finally, we assume the mean of each agent's starting inventory is the same within a given sub-population so that $\E[\mfQ^j_0] = \invmean^k_0$, $\forall\;j \in \mfK_k^{\N}$.

The amount of cash any agent has accumulated through their trading is represented via their (controlled) cash process $\{(X_t^{j,\nuj})_{t\in[0,T]}\}_{j\in\mfN}$. When buying or selling the asset, we assume each trader pays an instantaneous transaction cost that is linearly proportional to the amount of shares transacted. This instantaneous cost is expressed through the controlled dynamics of the cash process which, for agent $j$, is given by
\begin{equation}
\label{eqn:Xdef}
	X_t^{j,\nuj} = X_0^j - \int_0^t  \left( S_u^{\nubarN} - a_k \, \nuj_u \right) \nuj_u\,du
	\;,\qquad \forall j\in\mfK_{k}^{\N}\,,
\end{equation}
where $a_k>0$ is a parameter unique to sub-population $k$, and $S^{\nubarN} = ( S_t^{\nubarN} )_{t\in[0,T]}$ is the (controlled) price process of the traded asset. We assume the midprice process $S_t^{\nubarN}$ can be written as
\begin{equation} \label{eq:S-price-dynamics}
	S_t^{\nubarN} = S_0 + \int_0^t \left( \lambda \, \nubarN_u + A_u \right) \, du + M_t
	\;,
\end{equation}
where $A=\left(A_t\right)_{t\in[0,T]}$ is a $\G$-predictable process, $M=\left( M_t \right)_{t\in[0,T]}$ is a $\G$-adapted martingale with $M_0=0$, and $\nubarN_u= \tfrac{1}{N} \sum_{j=1}^N \nu_u^j$ is the average trading rate of all agents. Additionally, we make the technical assumption that $A\in\HT$ and $M\in\LT$, where
\begin{equation}
	\HT =
	\left\{ \nu:\Omega\times[0,T] \rightarrow \R \,,\; \E\left[\textstyle\int_0^T \nu^2_u\,du \right] < \infty \right\}
\end{equation}
is the set of $\mathbb P$-square integrable processes on $[0,T]$, and where
\begin{equation}
	\LT =
	\left\{ \nu:\Omega\times[0,T] \rightarrow \R \,,\; \E\left[ \nu^2_t \right]  < \infty  \;, \forall \; t \in [0,T] \right\}
\end{equation}
is the set of processes with finite $\mathbb P$-second-moment on the interval $[0,T]$. We also assume the components driving the price process $A$ and $M$ are independent of $\{ \mfQ_0^j \}_{j=1}^N$, the initial values of the agent's inventories and that the quantity $\E^{\Pk} \left[ \int_0^T A_u \, du \right]$ is invariant to any agent's choice of trading strategy $\nuj\in\HT$.

Note that, beyond the integrability conditions on $A$ and $M$ there are no a-priori specifications on the dynamics of these processes. We allow enough flexibility so that $S$ may have very general semi-martingale dynamics which may incorporate jumps or even be non-Markovian. The only binding assumption is that order-flow from agents' trading have linear price impact. Having such a general class of models allows us to obtain optimal controls that are robust to specific model choices.

The process $A$ characterizes the `alpha' or mean trajectory of the asset price process, while $M$ represents the noise surrounding this drift, and both can be thought of as incorporating the effects of information sources not available to the agents, i.e., they can be latent. The parameter $\lambda$ controls the scale of the (permanent) price impact of the agent's strategies.
The average population price impact $\nubar_t^{\N}$ admits an alternative representation in terms of averages over sub-populations,
\begin{equation} \label{eq:average-impact-def}
	\nubarN_t = \sum_{k\in\mfK} \tfrac{N_k^{\N}}{N} \,  \nubar_t^{k,\N}\,,
	\quad\text{where }\quad
	\nubar_t^{k,\N} := \tfrac{1}{N_k^{\N}} \sum_{j\in\mfK_k^{\N}} \nuj_t
	\;.
\end{equation}
This representation will be useful in the infinite population size limit studied in Section~\ref{sec:solving-the-mfg}.

\subsection{Information Restriction} \label{sec:information-restriction}
In our model, we restrict the flow of information to the agents by allowing agent $j$ to only have access to information about the paths of the price process $S_t^{\nubarN}$ and their own inventory, but not others.
More explicitly, we allow an agent $j$ to choose a control $\nuj$ from the space of admissible controls (for agent-$j$)
\begin{equation}
	\A^j := \left\{ \nu\text{ is }\F^j\text{-predictable}\,,\; \nu \in \HT \right\}
	\;,
\end{equation}
where for each $j\in\mfN$ we define the filtration $\F^j=(\F^j_t)_{t\in[0,T]}$, where
\[
\F_t^j=\sigma\left( (S_u^{\nubarN})_{u\in[0,t]} \right) \vee \sigma\left( \mfQ_0^j \right)
\]
is the sigma algebra generated by the paths of the asset midprice process and agent $j$'s starting inventory level\footnote{Note: Although we do not consider it in this paper, the definition of $\F$ can just as easily be generated to include paths of additional sources of information to the agents.}.  If $\nu\in\A^j$, then $X_t^{j,\nu}$ and $\qj{\nu}_t$ are $\F^j$-adapted as well, therefore we omit the sigma algebras generated by these processes from the definition of $\A^j$ and $\F^j$. Let us also define the filtration $\F=(\F_t)_{0\leq t\leq T}$, such that $\F_t=\sigma\left( (S_u^{\nubarN})_{0\leq u\leq t} \right)$, to be the filtration generated solely by the paths of the asset midprice process.

To understand the implications of the above restriction, we can take a close look at the dynamics of $S_t^{\nubarN}$ in equation~\eqref{eq:S-price-dynamics}. Firstly, since agent-$j$ is restricted to information in $\F^j$, they will be unable to observe the $\G$-adapted components $A$ and $M$, and they are also unable to observe the inventory levels of other agents. Consequently, they are unable to observe the trading strategies of all other agents. Instead, the trader must reconstruct the values of each of these individual components only through observations of the paths of $S_t^{\nubarN}$. Both $A$ and $M$ can be regarded as latent processes which, potentially, must be estimated to account for the effect of price movements.

We now provide a simple example that illustrates this framework in the single agent case, along the lines of the latent alpha models that \cite{casgrain_jaimungal_2016} study. To this end, assume there is a  latent Markov chain $\Theta = (\Theta_t)_{t\in[0,T]}$ unobservable to any agent, and further assume that
\begin{equation}
	S_t^{\nubarN} = S_0 + \lambda \int_0^t \nubar_u^{\N} \,du + \int_0^t f(\Theta_u)\,du + \sigma \, W_t
	\;.
\end{equation}
In this example, the large filtration $\G$ is given by $\G_t=\sigma\left( (\Theta_t,W_t)_{t\in[0,T]} \right) \vee \sigma( \{\mfQ^j_0 \}_{j=1}^N )$. By restricting agent-$j$ to trade on the filtration $\F^j$, they may only trade based on the observed path of the asset price and not its individual components. In an algorithmic trading setting, it is important for agents to have a predictive model for the asset price process. Since the dynamics of $S^{\nubarN}$ are not fully known to the agents, they will have to infer values of $\Theta_t$ based on the filtration at time $t$ in order to make predictions on the future value of the asset price.

Our set-up allows for more general latent models than this one example, but it is useful to keep this example in mind when thinking about a concrete case. We refer the reader to \cite{casgrain_jaimungal_2016} for more details on latent alpha models such as the example presented here.

\subsection{The Agent's Optimization Problem}

Each agent wishes to maximize an objective functional which measures their trading performance over the trading period $[0,T]$. For each $j\in\mfN$, let $\A^{-j}:=\bigtimes_{i\in\mfN, i\ne j} \A^i$, we assume that agent-$j$ within a sub-population $k\in\mfK$ chooses a control $\nuj\in\A^j$ to maximize a functional $H_j:\A^j\times\A^{-j}\rightarrow\R$, which is defined as
\begin{equation} \label{eq:sub-pop-objective-rep}
	H_j(\nuj,\nu^{-j})=\E \left[
	X_T^{\nuj}
	+ \qj{\nuj}_T \left( S_T^{\nubarN} - \Psi_k\, \qj{\nuj}_T \right)
	- \phi_k \int_0^T \left( \qj{\nuj}_u \right)^2 \, du
	\right]
	\;,
\end{equation}
where $\phi_k\geq0$, $\Psi_k>0$ are constants that may vary by sub-population $k$, and where $\nu^{-j}:=\left( \nu^1 , \dots , \nu^{j-1},\nu^{j+1}, \dots, \nu^N \right)$ indicates the dependence of the objective on the controls of all other agents.

The agent's objective is composed of three distinct parts. The first component, $X_T^{\nuj}$, is the agent's total accumulated cash. The second component, $\qj{\nuj}_T \left( S_T^{\nubarN} - \Psi_k \, \qj{\nuj}_T \right)$, represents the mark-to-market value of the terminal inventory and includes a liquidation penalty. Indeed, as $\Psi_k\rightarrow\infty$, the agent completely liquidates their inventory by the end of the trading interval, since, in this limit, the cost of holding non-zero inventory at time $T$ goes to infinity. Lastly, $- \phi_k \int_0^T \left( \qj{\nuj}_u \right)^2 \, du$ represents a running penalty that penalizes an agent for holding large long or short inventory positions throughout the trading period. The parameter $\phi_k$ can be regarded as controlling the risk appetite of the agent, since for large values of $\phi_k$, agent $j$ will have a great dis-incentive to take on any market exposure. This penalty can also be understood from the perspective of the agent accounting for model uncertainty as analysed in \cite{cartea2017algorithmic}.

Agents interact through the price impact term $\nubarN$, which appears implicitly in the dynamics of $S^{\nubarN}$. Agents within the same sub-population have the same objective functional, this implies that sub-populations act in a similar manner, albeit each individual agent's strategy is adapted to their own inventory (in addition to the midprice), and hence agents' strategies are not identical.

Substituting $\qj{\nuj}_T$, $X_T^j$, and $S_T$ using \eqref{eq:agent-inventory-dynamics}, \eqref{eqn:Xdef}, and \eqref{eq:S-price-dynamics}, respectively into definition~\eqref{eq:sub-pop-objective-rep}, using integration by parts, and taking expectations, we obtain the alternative form of the objective functional:
\begin{equation} \label{eq:alternative-objective-def}
\begin{aligned}
	H_j(\nuj,\nu^{-j})
	= \E\Bigg[
	\; X_0^j &+ \mfQ^j_0 \left( S_0 - \Psi \mfQ_0^j \right)
	\\
	&+ \int_0^T
	\left\{
	\qj{\nuj}_t dS_t^{\nubarN}
	-
	\begin{pmatrix}
		\nuj_t \\ \qj{\nuj}_t
	\end{pmatrix}^\T
	\begin{pmatrix}
	a_k & \Psi_k \\ \Psi_k & \phi_k
	\end{pmatrix}
	\begin{pmatrix}
		\nuj_t \\ \qj{\nuj}_t
	\end{pmatrix}
	\, dt
	\right\}
	\Bigg]
	\;.
\end{aligned}
\end{equation}
This representation makes the influence of the parameter triplet $(a_k,\phi_k,\Psi_k)$ on the objective function explicit. The triplet, which is shared amongst all members of sub-population $k$, will have a direct impact on the agent's behaviour. The variation of this triplet across sub-populations allows us to incorporate heterogeneous agents.

\begin{remark}
	The model presented above is designed so that agents are incentivized to gradually bring their inventory levels $\qj{}_t$ towards zero over the course of the trading period. In particular, an agent-$j$ in sub-population $k$ is penalized for non-zero exposure by the terminal liquidation penalty $- \Psi_k\, (\qj{\nuj}_T)^2$ and the running penalty $-\phi_k \int_0^T (\qj{\nuj}_t)^2 \, dt$ which both appear in the expression for the objective function \eqref{eq:sub-pop-objective-rep}.

	It is possible to generalize this model to instead pressure agents to bring their inventory levels towards some stochastic trading target $\{\mfQ^j_T \}_{j=1}^N$, so that agent-$j$ is instead penalized for deviating from its trading target $\mfQ^j_T$ at time $T$, rather than deviating from $0$. This can be achieved by replacing the former terminal liquidation and running penalties with $-\Psi_k (\qj{\nuj}_T - \mfQ^j_T)^2$ and $-\phi_k \int_0^T (\qj{\nuj}_t - \mfQ^j_t)^2  dt$, respectively. Because of the linear structure in the midprice model, it is easy to show that this modification to the objective function is exactly equivalent to modifying the initial condition of the process $\qj{}$ from $\qj{\nuj}_0=\mfQ^j_0$ to $\qj{\nuj}_0=\mfQ^j_0 - \mfQ^j_T $ for each $j\in\mfN$. If we impose the conditions that $\{\mfQ^j_0 - \mfQ^j_T\}_{j=1}^T$ are independent, that $\E (\mfQ^j_0 - \mfQ^j_T)^2 < C \; \forall \; j\in\mfN$, and that $\E[\mfQ^j_0 - \mfQ^j_T] = \invmean^k_0$ $\forall\;j \in \mfK_k^{\N}$, then all of the results that follow in the remainder of the paper apply for the generalized model with stochastic trading targets, which is done by simply replacing the initial condition of each inventory process with $\mfQ^j_0 - \mfQ^j_T$.

	\end{remark}

\subsection{The Stochastic Game}

As mentioned earlier, all agents seek to maximize their own objective function, and we seek the optimal strategy for all agents. More formally, we seek a collection of controls ${ \{ \omega^j \in \A^j \, : j\in\mfN\} }$ such that
\begin{equation}
	\omega^j = \arg\sup_{\omega\in\A^j} H_j(\omega,\omega^{-j})\;, \qquad \forall j\in\mfN\,.
\end{equation}
Identifying this collection of controls is no easy feat, since the objective functional for each agent is affected by the controls of all agents. Furthermore, the set of admissible controls differs between agents -- recall each agent has access to the filtration generated by the midprice and their inventory only. This latter observation posses difficulties since the set of optimal controls we are searching for, and the random processes present in the objective function \eqref{eq:sub-pop-objective-rep}, are adapted to different filtrations. This fact prevents us from (directly) applying the standard set of dynamic programming or stochastic maximum principle  tools to solve the problem.

\section{Solving the Mean-Field Stochastic Game} \label{sec:solving-the-mfg}

The stochastic game we aim to solve presents a number of obstacles which prohibit it from being solved directly. In this section, we overcome these obstacles by instead solving a MFG version of the stochastic game. To construct the MFG, we take the limit as the population size tends to infinity. In the limit, the finite player game becomes a (stochastic) MFG where agents no longer interact directly with one another, but instead interact through a set of mean-field processes (one for each sub-population). In the remainder of this section, we present the MFG that results from the infinite population limit, provide a closed-form representation of each agent's optimal strategy, and a closed-form representation of the mean-field processes within the game. Although we do not explicitly solve the finite player game presented in Section~\ref{sec:Model}, by establishing an $\epsilon$-Nash equilibrium property in Section~\ref{sec:epsilon-Nash}, we show that the equilibrium solution obtained for the MFG provides an good approximation to the finite population game, provided that the population size is large enough.

We begin the section by taking the population limit of the stochastic game, and of each agent's objective functional as $N\rightarrow\infty$, to yield a new limiting objective function and a (stochastic) MFG. Focusing on this limiting objective functional, we proceed by applying tools from convex analysis to obtain the optimal action for each agent in the form of the solution to a vector-valued forward-backward stochastic differential equation (FBSDE). Next, we obtain the equilibrium in the MFG by explicitly solving the FBSDE.


\subsection{The Limiting Mean-Field Game}

From \eqref{eq:alternative-objective-def}, we see the only term that depends on the population size $N$, within each objective functional, is the average trading rate of all agents $(\nubarN_t)_{t\in[0,T]}$. Moreover, this dependence appears only through the midprice dynamics $S^{\nubarN}$.

To formulate the limiting problem, we make some additional assumptions regarding the existence of a limit of the average trading rate. Let us assume that there exist processes $\nubark=(\nubark_t)_{t\in[0,T]}$ for $k \in\mfK$ so that $\nubark \in \HT$ and $\nubark$ is $\F$-predictable, where
\begin{equation}
	\lim_{N\nearrow\infty} \nubarkN_t = \nubark_t
	\, , \quad \P \times \mu \text{ a.e.}
	\; ,
\end{equation}
$\mu$ is the Lebesgue measure on the Borel sigma algebra on $[0,T]$, and $\P \times \mu$ is the canonical product measure of $\P$ and $\mu$.
\begin{remark}
The assumption that each $\nubark$ is $\F$-predictable can be relaxed to being predictable w.r.t to the finer filtration $\vee_{j\in\mfN}\F^j$ without any change to the results that follow.
\end{remark}
We call each of the processes $\{\nubark\}_{k\in\mfK}$ the sub-population mean-fields, where each component represents the limiting average trading rate within a given sub-population. Due to the assumptions on the  relative size of the sub-populations (see \eqref{eq:limit-assumption-proportions}),  in the limit as $N\rightarrow\infty$, the total average rate of trading $\nubarN$ exists. More specifically, let us define the population mean-field to be the process $\nubar = (\nubar_t)_{t\in[0,T]}$ where $\nubar_t = \lim_{N\nearrow\infty} \nubarN_t$, then $\nubar_t$ exists and admits the representation
\begin{equation} \label{eq:nubar-subpop-rep}
	\nubar_t = \sum_{k\in\mfK} p_k \, \nubark_t
	\,, \qquad \P\times\mu \text{ a.e.}
	\;,
\end{equation}
where $\{p_k\}_{k\in\mfK}$ represent the limiting proportions of each sub-population as defined in equation~\eqref{eq:limit-assumption-proportions}. With these assumptions, the limiting dynamics of the asset price process satisfies the (controlled) SDE
\begin{equation}
	dS_t^{\nubar} =
	\left( A_t + \lambda \sum_{k\in\mfK} p_k \, \nubark_t \right) \, dt
	+ dM_t
	\;,
\end{equation}
where we replace all of the price impact terms with their limits.

Since we restrict agent-$j$ to trading actions $\nuj$ from the admissible set $\A^j \subset \HT$, as $N\rightarrow\infty$, each agent's individual contribution to $\nubar^{\N}_t$ vanishes. Furthermore, upon inspection of the definition of agent-$j$'s objective functional in equation~\eqref{eq:alternative-objective-def}, we see the dependence on $\nu^{-j}$ appears only through the process $\nubarN$, which converges to $\nubar_t$ in the limit. These two remarks imply that as $N\rightarrow\infty$, each agent's objective functional no longer depends directly on $\nu^{-j}$, but rather it depends on the population statistic $\nubar_t$, representing the average price impact of all agents in the limit. For ease of notation, we therefore suppress the second argument of the objective functional in the limit. Hence, an agent $j$ in sub-population $k$, seeks to maximize the functional $\Hbar_j : \A^j \rightarrow \R$,
\begin{equation} \label{eq:Hbar-obj-def}
	\Hbar_j(\nuj)
	= \E\Bigg[
	\int_0^T
	\left\{
	\qj{\nuj}_t
	dS_t^{\nubar}
	-
	\begin{pmatrix}
		\nuj_t \\ \qj{\nuj}_t
	\end{pmatrix}^\T
	\begin{pmatrix}
	a_k & \Psi_k \\ \Psi_k & \phi_k
	\end{pmatrix}
	\begin{pmatrix}
		\nuj_t \\ \qj{\nuj}_t
	\end{pmatrix}
	\right\}
	\, dt
	\Bigg]
	\;,
\end{equation}
which we obtain using representation~\eqref{eq:alternative-objective-def} and omitting the constant terms. This above objective function implicitly depends on the processes $\{\nubark\}_{k=1}^K$, which we will need to determine when proceeding with the agent's optimization problem.

We next aim to solve the mean-field stochastic game by identifying a set of strategies which form a Nash equilibrium. In other words, we seek a collection of controls $\{ \nuj \}_{j=1}^\infty$ so that
\begin{equation}
	\nuj = \arg \sup_{\nuj \in \A^j} \Hbar_j (\nu)
	\;,
\end{equation}
for all $j\in\mfN$. Because of the definition of $\nubar_t$, we require the collection of controls to simultaneously satisfy the consistency condition
\begin{equation} \label{eq:consistency-condition}
	\nubar_t = \lim_{N \nearrow \infty} \frac{1}{N} \sum_{j=1}^N \nuj_t
	\;.
\end{equation}
This optimal control problem is of a similar form to the discrete population game, with the main difference lying in the consistency condition imposed on the collection of controls at the optimum.

The model we use here falls into the class of \emph{extended} mean-field games of controls, such as in in~\cite{gomes2014existence,carmona2018probabilistic}. The particular `fixed-point' formulation of the mean-field which appears through the consistency condition~\eqref{eq:consistency-condition} follows the approach to mean-field games of~\cite{HuangCaines_TAC07,nourian2013e,firoozi2015varepsilon}, amongst others.

\subsection{Solving the Agent's Optimization Problem}

In this section, we solve for the collection of controls that form a Nash equilibrium for the MFG. To achieve this, we use techniques from the convex analysis literature in a similar in spirit to the approach used in~\cite{bank2017hedging} for the problem of optimal hedging. We conclude by demonstrating that the agent's optimal control can be represented as the solution to a particular linear vector-valued FBSDE.

Solving for an optimal control for objective function $\Hbar_j$, defined in~\eqref{eq:Hbar-obj-def}, presents some challenges due to the latent information structure, stemming from the agents inability to observe the individual components of the midprice process. Each agent aims to find an $\F^j$-adapted control to maximize an objective function containing costs that are adapted to the filtration $\G\supseteq\F^j$. These $\G$ adapted processes appear solely in the dynamics of the midprice process $S_t^{\nubar}$. Due to this latent information, it is not possible to directly apply standard stochastic control techniques to obtain the agent's optimal behaviour.

Instead of taking a direct approach, we first represent the midprice process in terms of a pair of $\F$-adapted processes. This will allow us to re-write $\Hbar_j$ in terms of processes that are $\F^j$-adapted, thus resolving the issue of latent information. This can be achieved by applying the result in the following lemma.
\begin{lemma} \label{thm:predictable-representation}
Define the process $\Ahat=(\Ahat_t)_{t\in[0,T]}$ where $\Ahat_t = \E\left[ A_t \mid \F_t \right]$. Then $\Ahat$ is an $\HT$,  $\F$-adapted process. Furthermore, there exists an $\F$-adapted $\LT$ martingale $\Mhat = ( \Mhat_t )_{t\in[0,T]}$ such that
	\begin{equation} \label{eq:SinTermsOfMhat}
		S_t^{\nubar} = S_0+\int_0^t \left( \Ahat_u + \lambda \, \nubar_u \right)\, du + \Mhat_u
		\;.
	\end{equation}
	The process $\Mhat$ is known as the innovations process for the filter $\Ahat$.
\end{lemma}
\begin{proof}
	The proof is found in \ref{sec:proof-thm-predictable-representation}
\end{proof}

Since $A$ and $M$ are independent of the agent's initial inventories $\{\mfQ_0^j\}_{j=1}^\infty$, the projection of the dynamics of $S^{\nubar}$ onto the filtration $\F$ will be identical to the projection onto $\F^j$. Moreover, since $\F_t \subseteq \F_t^j$, the processes $\Ahat$ and $\Mhat$ are adapted to the each individual agent's filtration. Lemma~\ref{thm:predictable-representation} provides us with a representation of $S^{\nubar}$ in terms of $\F$-adapted processes rather than the $\G$-adapted version in \ref{eq:S-price-dynamics}. Plugging in these $\F$-dynamics into the expression for $\Hbar_j$ in~\eqref{eq:alternative-objective-def}, and noticing that the martingale terms vanish under the expectation, we obtain an objective functional entirely in terms of $\F^j$-adapted processes,
\begin{equation} \label{eq:F-adapted-objective}
	\Hbar_j(\nuj)
	= \E\Bigg[
	\int_0^T
	\left\{
	\qj{\nuj}_t
	\left(
	\Ahat_t + \lambda \, \nubar_t
	\right)
	-
	\begin{pmatrix}
		\nuj_t \\ \qj{\nuj}_t
	\end{pmatrix}^\T
	\begin{pmatrix}
	a_k & \Psi_k \\ \Psi_k & \phi_k
	\end{pmatrix}
	\begin{pmatrix}
		\nuj_t \\ \qj{\nuj}_t
	\end{pmatrix}
	\right\}
	\, dt
		\Bigg]
	\;.
\end{equation}

In this representation, the partial information problem is cast into a full information problem, and we can now apply convex analysis tools to this objective functional. The essence of the steps used to obtain the solution resembles very closely those used in elementary calculus to find critical points of functions. First, show that the objective function is `differentiable' and strictly concave. Since the argument of $\Hbar_j$ is a stochastic process, we refer to `differentiability' in the sense of the G\^ateaux directional derivative\footnote{For more information on the G\^ateaux derivative and its role in convex optimization, see \cite[Section 5]{ekeland1999convex}.}. Next, identify where the G\^ateaux derivative vanishes to characterize the objective functional's critical points. Finally, knowing that the objective function is strictly concave, guarantees that the critical point is unique and that it is a maximum. The lemmas that follow show that $\Hbar^j$ is both concave and everywhere G\^ateaux differentiable in $\A^j$.

\begin{lemma} \label{lemma:obj-is-concave}
	The functional $\Hbar_j$ defined in equation~\eqref{eq:sub-pop-objective-rep} is  strictly concave in $\A^j$.
\end{lemma}
\begin{proof}
	The proof is found in \ref{sec:appendix-proof-obj-is-concave}.
\end{proof}

\begin{lemma} \label{lemma:GateauxDeriv}
The objective function $\Hbar_j$ is everywhere G\^ateaux differentiable in $\A^j$. Its G\^ateaux derivative at a point $\nu\in\A^j$ in a direction $\omega\in\A^j$ can be expressed as
\begin{equation} \label{eq:j-gateax-deriv}
\begin{split}
	\left\langle \D\Hbar_j(\nu) , \omega \right\rangle
	=
	\E\left[
	\int_0^T
	\omega_t
	\;
	\E\Bigg[ \right.
&
	-2 \,a_k \,\nu_t - 2 \,\Psi_k \,\qj{\nu}_T \\
& \left.	
+ \left.\int_t^T \left\{ \Ahat_u + \lambda \,\nubar_u - 2\,\phi_k \,\qj{\nu}_u  \right\} \, du \; \right\lvert \F_t^j  \Bigg]
	 \, dt \,
	\right].
\end{split}
\end{equation}
\end{lemma}
\begin{proof}
	The proof is found in \ref{sec:appendix-proof-lemma-differentiable}.
\end{proof}

Since the objective functional $\Hbar_j$ is concave and G\^ateaux differentiable, an element $\nu\in\A^j$ that makes the G\^ateax derivative vanish in an arbitrary direction $\omega\in\A^j$ is guaranteed to be a maximizer. Moreover, since the concavity of $\Hbar_j$ is strict, the maximizer is unique. The explicit form of the derivative in expression~\eqref{eq:j-gateax-deriv} allows us to find a representation of the agent's optimal strategy. The following proposition uses these last two results to represent the trader's optimal strategy as the solution to an FBSDE.

\begin{proposition} \label{thm:optim-bsde-prop}
The collection of controls $\left\{\nujst\right\}_{j\in\mfN}$ forms a Nash equlibrium if and only if for each agent-$j$ in sub-population $k$, $\nujst\in\HT$ and $\nujst$ is the unique strong solution to the FBSDE
\begin{equation} \label{eq:thm-sol-statement-2}
	\begin{cases}
		&-d(2\,a_k\,\nujst_t)
		= \left( \Ahat_t + \lambda \,\nubar_t - 2\,\phi_k \,\qj{\nujst}_t \right) \,dt - d\mathcal{M}_t^j\,,
		\\
		&\hspace{1em}
		2\,a_k \, \nujst_T = -2\,\Psi_k \,\qj{\nujst}_T\,,
	\end{cases}
\end{equation}
where $\mathcal{M}^j\in\HT$ is an $\F^j$-adapted martingale and $\nubar_t := \displaystyle\lim_{N\rightarrow\infty} \tfrac{1}{N} \sum_{j=1}^\infty \nujst_t$.
\end{proposition}
\begin{proof}
	The proof is found in \ref{sec:proof-thm-optim-bsde-prop}.
\end{proof}

Equation~\eqref{eq:thm-sol-statement-2} is an FBSDE since it has a forward component coming from the processes $\qj{\nujst}$ and $\Ahat$, as well as a backward component $\nujst$, which must be solved simultaneously. A solution to equation~\eqref{eq:thm-sol-statement-2} is the unique optimal control for agent-$j$ and maximizes their objective functional.
Note that all agent's optimal strategies are coupled through these FBSDEs via the mean-field process $\nubar$ which appears in the driver of equation~\eqref{eq:thm-sol-statement-2}, and they must all satisfy the consistency condition. As well, the parameters in the performance criteria depend on the specific sub-population to which the agent belongs.

\subsection{Solving the Mean Field Equations}

In the MFG limit, the infinite dimensional stochastic game is reduced to solving the FBSDE~\eqref{eq:thm-sol-statement-2} system. The FBSDE for agent-$j$'s optimal control shows there is no direct dependence on any other individuals' choice of strategy. Instead, the effect of all other agents appear through the mean field process $\nubar$. Hence, rather than having the explicit dependence of one $\nu^{j,\ast}$ on ${\nu}^{-j,\ast}$, we have an implicit dependence on ${\nu}^{-j,\ast}$ through the mean field process $\nubar$. Furthermore, since the FBSDE~\eqref{eq:thm-sol-statement-2} depends on a particular agent's sub-population, it is necessary to separate the problem across sub-populations and solve for each of their mean-fields and optimal controls.

In the remainder of the section we solve the FBSDE~\eqref{eq:thm-sol-statement-2}. The main obstacle in solving the FBSDE is that the mean-field process $\nubar$ depends on the solution of each $\nuj$ and vice versa, through the driver of the FBSDE. To overcome this obstacle, we first decompose the mean-field process as the average of the sub-population mean fields, i.e., write $\nubar_t=\sum_{k\in\mfK} p_k \, \nubark_t$. Next, we formulate an ansatz for each $\nubark_t$, which we then use to find a corresponding ansatz of the solution to $\nuj$. We then conclude by demonstrating that the ansatz solutions for both $\{\nuj\}_{j\in\mathbb{N}}$ and $\{\nubark\}_{k\in\mfK}$ indeed form the unique solution to the FBSDE problem~\eqref{eq:thm-sol-statement-2}.

\subsubsection{An Ansatz for the Mean-Field Processes} \label{sec:Mean-Field-Ansatz}

Our first task is to propose an appropriate form of each sub-population mean-field process $\nubark$. Our proposed ansatz for each sub-population mean-field  is denoted by the process $\nutk = (\nutk_t)_{t\in[0,T]}$, $\forall\;k\in\mfK$, and solves the FBSDE
\begin{equation} \label{eq:mean-field-ansatz-FBSDE}
\left\{
	\begin{array}{rl}
			-d(2 \,a_k\,\nutk_t)
			&\!\!= \left( \Ahat_t + \lambda \sum_{k^{\prime}\in\mfK} p_{k^\prime} \, \nut^{k^\prime}_t - 2\,\phi_k \,\qtk{\nutk}_t \right) dt - d\mMbar_t^k\,,
			\\
			2 \,a_k \, \nutk_T &\!\!= -2\,\Psi_k \,\qtk{\nutk}_T\,,
	\end{array}
\right.
\end{equation}
where the (controlled) forward processes $\qt^{k,\nutk}=(\qt_t^{k,\nutk})_{t\in[0,T]}$ are given by $\qtk{\nutk}_t = \invmean_0^k + \int_0^t \nutk_u \, du$, $\invmean_0^k = \E[\mfQ_0^j]$, and where $\mMbar^k=(\mMbar^k_t)_{t\in[0,T]}$ is a suitable $\F$-adapted martingale satisfying $\mMbar^k\in\HT$. It is worth pointing out that the martingales appearing here are all $\F$-adapted, and not $\F^j$-adapted as in \eqref{eq:thm-sol-statement-2}. We will nonetheless see that the ansatz does indeed provide a solution to our original problem.

The FBSDEs in \eqref{eq:mean-field-ansatz-FBSDE} can be viewed as resulting from taking the average over all $j\in\mfK_k$ of the FBSDEs~\eqref{eq:thm-sol-statement-2}, and explicitly splitting the overall mean-field in terms of the sub-population mean-fields.
Since the system is an average over an infinite number of objects, there is no guarantee that the solution to this `average' FBSDE will exactly match the sub-population mean-field process $\nubark$. We show in Theorem \ref{thm:mean-field-is-true} that the solution to \eqref{eq:mean-field-ansatz-FBSDE} does indeed provide us with the mean field process $\nubark_t = \lim_{N\nearrow\infty} \frac{1}{N_k^{\N}}\sum_{j\in\mfK_k^{\N}}\nu_t^j$, where $\nuj$ is the solution to the FBSDE~\eqref{eq:thm-sol-statement-2}.

We now solve the collection FBSDEs~\eqref{eq:mean-field-ansatz-FBSDE} using an approach similar to (but not the same as) the `four-step method' of \cite{ma1994solving}. The linear structure of the coupling of the FBSDEs plays a key role, and we re-write the collection as a single vector-valued equation. First, let $\bnut_t = \left( \nutk_t \right)_{k\in\mfK}$ be the column vector of the ansatzes for each sub-population mean-field. Stacking each of the mean-field FBSDE~\eqref{eq:mean-field-ansatz-FBSDE} results in the vector valued equation,
\begin{equation} \label{eq:stacked-FBSDE}
	\begin{cases}
			&-d(2 \, \ba \,\bnut_t)
			= \left( \Ahat_t \, \bm{1}^{(K\times 1)} + \bLambda \, \bnut_t - 2 \, \bphi \, \bqt_t \right) \,dt - d\bmMbar_t\,,
			\\
			&\hspace{1em}
			2 \, \ba \,\bnut_T = - 2 \, \bPsi \, \bqt_T\,,
		\end{cases}
\end{equation}
where $\ba,\bphi,\bPsi$ and $\bLambda$ are all real-valued $K \times K$ matrices defined as
\begin{align*}
	\ba &= \text{diag}\left( \{a_k\}_{k\in\mfK}\right) \,, \hspace{3em}
	\bphi = \text{diag}\left( \{\phi_k\}_{k\in\mfK}\right)\,,  \\ \vspace{5cm}
	\bPsi &= \text{diag}\left( \{\Psi_k\}_{k\in\mfK}\right) \,, \hspace{3em}
	\bLambda =
	\begin{pmatrix}
		\lambda \, p_1 & \dots & \lambda \, p_K \\
		\vdots &  & \vdots \\
		\lambda \, p_1 & \dots & \, \lambda p_K
	\end{pmatrix}
	\;,
\end{align*}
$\bqt_t = \boldsymbol{\invmean}_0 + \int_0^t \bnut_u \, du$, and  $\bmMbar_t = ({\mMbar}^k_t)_{k=1}^K$ is a column vector of $\F$-adapted martingales with $\mMbar_t^k \in \HT$, $\forall k\in\mfK$.

Due to the linear structure of the vector-value FBSDE~\eqref{eq:stacked-FBSDE}, we make the further ansatz that there are two $\F$-adapted processes $\g_1 = ( \bg_{1,t} )_{t\in[0,T]}$ and $\g_2 = ( \bg_{2,t} )_{t\in[0,T]}$, where $\bg_{1,t}\in\R^{K}$ and $\bg_{2,t}\in\R^{K\times K}$, such that the solution to \eqref{eq:stacked-FBSDE} can be expressed as
\begin{equation}
	2 \,\ba \,\bnut_t = \bg_{1,t} + \bg_{2,t} \,\bqt_t
	\;.
\end{equation}
Applying It\^o's lemma to the above expression, inserting the result back into \eqref{eq:stacked-FBSDE}, and grouping terms by $\bqt$, yields
\begin{equation} \label{eq:ansatz-dynamics}
	\begin{aligned}
	0 =&\phantom{+~}
	\Big\{
	d\bg_{1,t} + \left(
	\bm{1}^{(K\times 1)}\Ahat_t
	+ \left( \bLambda + \bg_{2,t} \right) \, \left( 2\,\ba  \right)^{-1} \, \bm{g}_{1,t}
	\right) dt
	- d\bmMbar_t
	\Big\}
	\\ &+
	\Big\{
	d\bg_{2,t} + \left( \left( \bLambda + \bg_{2,t} \right) \left( 2 \ba \right)^{-1} \bm{g}_{2,t}
	- 2\bm{\phi} \right) \, dt
	\Big\} \, \bqt_t
	\;.
	\end{aligned}
\end{equation}
Equation~\eqref{eq:ansatz-dynamics} must hold $\P\times\mu$ almost everywhere for all $\bqt_t$, hence, the terms within each curly brace vanish independently. Moreover, we can apply the same argument to the boundary condition of $\bnut$ to yield two coupled BSDEs for $\bg_{1}$ and $\bg_{2}$ which no longer depend on the forward process $\bqt_t$. The first of these is a linear BSDE for $\bg_{1}$,
\begin{equation}
	\label{eq:linear-BSDE}
	\begin{cases}
		-d\hspace{-1em}&\bm{g}_{1,t} =
	\left(
	\bm{1}^{(K\times 1)}\Ahat_t
	+ \left( \bLambda + \bg_{2,t} \right) \, \left( 2\,\ba  \right)^{-1} \, \bm{g}_{1,t}
	\right) dt
		- d\bmMbar_t\,,
		\\
		&\bm{g}_{1,T} = \bm{0}^{(K \times 1)}\,,
	\end{cases}
\end{equation}
and a matrix-valued ODE for the value for $\bg_2$,
\begin{equation}
	\label{eq:riccati-ODE}
	\hspace{-6.5em}
	\begin{cases}
	-d\hspace{-1em}&\bm{g}_{2,t} = \left(
	 \left( \bLambda + \bg_{2,t} \right) \left( 2 \ba \right)^{-1} \bm{g}_{2,t}
	- 2\bm{\phi} \right) dt\,,
	\\
	&\bm{g}_{2,T} = -2\bPsi\,.
\end{cases}
\end{equation}

There is a one-way dependence structure in these equations. Equation~\eqref{eq:linear-BSDE} is a linear BSDE that depends on the solution for $\bg_2$, while equation~\eqref{eq:riccati-ODE} is a matrix-valued non-symmetric Riccati equation that is independent of $\bg_1$. Let us also note that equation~\eqref{eq:riccati-ODE} is an ordinary differential equation which is deterministic since it has no Martingale term and has a deterministic boundary condition. Such vector and matrix-valued BSDEs are reminiscent of those appearing in \cite{bouchard2017equilibrium}.
The solutions to \eqref{eq:linear-BSDE} and \eqref{eq:riccati-ODE} are presented in the proposition that follows.
\begin{proposition} \label{prop:solution-g1-g2}
	There exists a unique solution $\bg_{2,t}$ to the matrix valued ODE~\eqref{eq:riccati-ODE} that is bounded over the interval $[0,T]$.

Moreover, let $\bm{Y}_t:[0,T]\rightarrow\R^{2K\times K}$ be defined as
	\begin{equation}
		\bm{Y}_t =
		e^{(T-t)\bm{B}}
		\left( \bm{I}^{(K \times K)} ,\; -2\,\bPsi \right)^\T
		\;,
	\end{equation}
	where $\bm{B} \in \R^{2K \times 2K}$ is the block matrix
	\begin{equation}
		\bm{B} =
		\begin{pmatrix}
			\bm{0}^{(K\times K)} & -(2\ba)^{-1} \\
			-2\bphi & \bLambda (2\ba)^{-1}
		\end{pmatrix}
		\;.
	\end{equation}
If we define the matrix partition $\bm{Y}_t = \left( \bm{Y}_{1,t} , \bm{Y}_{2,t} \right)^\T$, where $\bm{Y}_{1,t},\bm{Y}_{2,t} \in \R^{K\times K}$, then $\bm{g}_{2,t}$ can be expressed as
	\begin{equation}\label{eqn:proposition-solution-g2}
		\bg_{2,t} = \bm{Y}_{2,t} \, \bm{Y}_{1,t}^{-1}
		\;.
	\end{equation}
	\\
Furthermore, the BSDE~\eqref{eq:linear-BSDE} admits a closed form solution,
	\begin{equation} \label{eq:proposition-solution-g1}
		\bg_{1,t} = \int_t^T
		\bm{:} e^{\int_t^u \, \left( \bLambda + \bg_{2,s} \right) \left( 2 \ba \right)^{-1} \, ds }\bm{:} \,
		\bm{1}^{(K\times 1)} \, \E \left[ A_u \lvert \F_t \right]
		\, du
		\;,
	\end{equation}
	where $\bm{:} e^{\int_t^u \cdot \, ds }\bm{:}$ represents the time-ordered exponential
	\footnote{We define the time-ordered exponential of a matrix-valued function $\bm{f}:[0,T]\rightarrow\R^{K\times K}$, $\bm{\eta}_{t,u} =\, {\bm{:} e^{\int_t^u \bm{f}_s \, ds }\bm{:}}$, with $t\leq u$, to be the unique solution to the matrix-valued ODE $d\, \bm{\eta}_{t,u} = \bm{\eta}_{t,u} \, \bm{f}_u \, du$ with the initial condition $\bm{\eta}_{t,u}=\bm{I}^{(K\times K)}$.}.
	Moreover, $\E[\int_0^T \bg_{1,t}^\T \, \bg_{1,t} \, dt] < \infty$.
\end{proposition}
\begin{proof}
	The proof is found in \ref{sec:proof-solution-g1-g2}.
\end{proof}

The processes $\bg_1$ and $\bg_2$ found in the above proposition provide us with a solution to the vector-valued FBSDE~\eqref{eq:stacked-FBSDE}. We summarize the results in the proposition that follows.

\begin{proposition} \label{thm:ansatz-prop-summary}
	Define the process $\bnut = (\bnut_t)_{t\in[0,T]}$, where $\bnut \in \R^{K}$ and
	\begin{equation}
		\bnut_t = ( 2\,\ba )^{-1} \left( \bg_{1,t} + \bg_{2,t} \; \bqt_t \right)
		\;.
	\end{equation}
	where $\bg_{1,t}$ and $\bg_{2,t}$ are the functions given in the statement of Proposition~\ref{prop:solution-g1-g2}, i.e., equations \eqref{eq:proposition-solution-g1} and \eqref{eqn:proposition-solution-g2}. Then $\bnut_t$ is the unique solution to the FBSDE~\eqref{eq:stacked-FBSDE}.

Furthermore, let $\nutk_t$ be the $k$th element of the vector $\bnut_t$, then for each $k\in\mfK$, $\nutk\in\bigcap_{j=1}^\infty \A^j$ and $\{\nutk_t\}_{k\in\mfK}$ form the solution to the collection of FBSDEs~\eqref{eq:mean-field-ansatz-FBSDE}.
\end{proposition}
\begin{proof}
	The proof is found in \ref{sec:proof-ansatz-prop-summary}.
\end{proof}

Furthermore, since $\bnut_t$ represents the ansatz vector for each sub-population mean-field, we may represent the total mean-field effect implied by the ansatz $\bnut_t$ as $\nut_t = \sum_{k\in\mfK} \, p_k \, \nut_t^k$.

\subsection{An Ansatz for Agent's Optimal Control}

We now use the ansatz derived in section~\ref{sec:Mean-Field-Ansatz} to derive an optimal control for each individual agent under the assumption that $\nubar_t = \nut_t$ (we show that indeed the solution is optimal in the next subsection). Consider the FBSDE~\eqref{eq:thm-sol-statement-2} for the optimal trading rate of an agent $j$ in sub-population $k$. Replacing the true mean field $\nubar_t$ with our ansatz $\nut_t$ in the optimality equation~\eqref{eq:thm-sol-statement-2}, we obtain the FBSDE
\begin{equation} \label{eq:agent-FBSDE-limit-mfansatz}
		\begin{cases}
			&-d(2\, a_k\,\nuj_t)
			= \left( \Ahat_t + \blambda^\T \, \bnut_t - 2\,\phi_k \,\qj{\nuj}_t \right) \,dt - d\mathcal{M}_t^j\,,
			\\
			&\hspace{1em}
			2 \,a_k \, \nuj_T = - 2\,\Psi_k \,\qj{\nuj}_T\,,
		\end{cases}
\end{equation}
where $\blambda = \left( \lambda p_k \right)_{k\in\mfK}$ and $\mathcal{M}_t^j$ is some square-integrable $\F^j$-adapted martingale. By solving this FBSDE, we find each agent's optimal control, assuming that the mean-field process is exactly equal to the ansatz derived in section~\ref{sec:Mean-Field-Ansatz}.


We solve the FBSDE \eqref{eq:agent-FBSDE-limit-mfansatz} along similar lines as the approach we took in solving \eqref{eq:stacked-FBSDE}. Doing so leads to the following proposition.

\begin{proposition} \label{prop:agents-optimal-control}
	Let agent$-j$ be a member of sub-population $k$, then the solution to the FBSDE~\eqref{eq:agent-FBSDE-limit-mfansatz} is
	\begin{equation} \label{eq:form-optimal-solution-mf-j}
		\nuj_t = \nut_t^k + \tfrac{1}{2 a_k} h_{2,t}^k
		\left( \qj{\nuj}_t - \qtk{\nutk}_t \right)
		\;,
	\end{equation}
	where $\nut_t^k$ is defined in Proposition~\ref{thm:ansatz-prop-summary} and $h_{2,t}^k:[0,T]\rightarrow \R$ is defined as
	\begin{equation}
		h_{2,t}^k = -2 \xi_k \left(
		\frac{
		\Psi_k \cosh\left( -\gamma_k (T-t)\ \right)
		- \xi_k \sinh\left( -\gamma_k (T-t)\ \right) }{
		 \xi_k \cosh\left( -\gamma_k (T-t)\ \right)
		- \Psi_k \sinh\left( -\gamma_k (T-t)\ \right)
		}
		\right)
		\;,
	\end{equation}
	where we the constants $\gamma_k = \sqrt{\phi_k/a_k}$ and $\xi_k=\sqrt{\phi_k a_k}$. Moreover, $h_{2,t}^k \leq 0$ for all $t\in[0,T]$, and $\nuj \in \A^j$.
\end{proposition}
\begin{proof}
The proof is found in \ref{sec:proof-agents-optimal-control}
\end{proof}

\subsection{Showing the Solution is Optimal}

At this point we have solved the optimality equation~\eqref{eq:thm-sol-statement-2} under the assumption that the mean-field $\nubar$ equals the ansatz mean-field process $\nut$. To show that the solution provided in Proposition~\ref{prop:agents-optimal-control} indeed solves the optimality equation, we need to demonstrate that $\nut$ is the true mean-field, i.e., that $\nut_t = \nubar_t = \lim_{N\rightarrow\infty} \tfrac{1}{N} \sum_{j=1}^N \nu_t^j$. To this end, we consider the error within a sub-population $k\in\mfK$, $\Delta^k_t = \nubark_t - \nutk_t$ and demonstrate that it equals zero. The result is summarized in the following theorem.
\begin{theorem} \label{thm:mean-field-is-true}
	Let $\nuj_t$ be the processes defined in Proposition~\ref{prop:agents-optimal-control}, and $\nutk_t$ be the processes defined in Proposition~\ref{thm:ansatz-prop-summary}. Then for each $k\in\mfK$
	 \begin{equation}
	 	\nut_t^k = \lim_{N\rightarrow\infty} \frac{1}{N_k^{\N}} \sum_{j\in\mfK_k^{\N}} \nuj_t\,,
	 	\qquad \P\times\mu \text{-almost everywhere}
	 	\;.
	 \end{equation}
	  Hence, the process
	  \begin{equation}
	  	\nujst_t = \nut_t^k + \tfrac{1}{2 a_k} h_{2,t}^k
		\left( \qj{\nujst}_t - \qtk{\nutk}_t \right)
	  \end{equation}
is the solution to the optimality equation~\eqref{eq:thm-sol-statement-2}, where $\nutk_t = \nubark_t$ $\P\times\mu$ a.e. and
	  \begin{equation}
	  	\nujst = \arg\sup_{\nu\in\A^j} \Hbar_j(\nu)
	  	\;.
	  \end{equation}
\end{theorem}
\begin{proof}
	The proof is found in~\ref{sec:proof-mean-field-is-true}.
\end{proof}

This theorem guarantees that the ansatz for the mean-field processes and each agent's control are indeed correct. Hence, for the MFG version of our stochastic game, we have closed-form solutions for the optimal strategy of each individual agent, the sub-population mean-fields, and the overall population mean-field.

\subsection{Properties of the Agent's Optimal Control} \label{sec:properties-of-the-optimal-control}

From Proposition \ref{thm:ansatz-prop-summary} (and Theorem \ref{thm:mean-field-is-true}), the mean-field trading rates within a sub-population can be written
\begin{equation} \label{eq:mean-field-structure}
	\bnubar_t = (2\ba)^{-1} \left( \,  \bg_{1,t} + \bg_{2,t} \, \bqbar_t \, \right)
	\;,
\end{equation}
where $\bqbar_t = \bm{\mu}_0 + \int_0^t \bnubar_t \, dt$.
The general structure of this mean-field strategy closely corresponds to the structure obtained for the single agent latent alpha model in \cite{casgrain_jaimungal_2016}.

The within sub-population mean-fields can be  decomposed into two parts. The first part, $( 2 \ba )^{-1} \, \bg_{1,t}$, represents the portion of the mean field trading rate that can be attributed to trading on alpha. This is evident from the representation of $\bg_1$ in equation~\eqref{eq:proposition-solution-g1}, which shows that $\bg_1$ is the weighted average of the expected future drift of the asset. Moreover, this expected future drift  is computed by conditioning on the agent's visible filtration only, meaning that the agent obtains the best possible estimate of the asset's alpha based on the information they have.

The second part, $( 2 \ba )^{-1} \, \bg_{2,t} \, \bqbar_t$, consists of a deterministic function multiplied by the vector of `mean-field inventories'. It admits the interpretation that the sub-population mean-field trading rates each induce the sub-population mean inventories towards zero (so that the terminal liquidation penalty is minimized), while simultaneously being conscious of the `mean-field inventories' of all other sub-populations.

%
%

From the result in Theorem \ref{thm:mean-field-is-true}, an agent $j$ in sub-population $k$, follows the strategy
\begin{equation} \label{eq:individual-optimal-control-real}
	\nujst_t = \nubark_t + \tfrac{1}{2a_k} h_{2,t}^k \left( \qj{\nujst}_t - \qbar_t^k \right)
	\,.
\end{equation}
Thus, each agent trades at the sub-population mean-field rate plus a correction term. Since $h_{2,t}^k$ is strictly negative, this correction term tends to push the agent's inventory towards the sub-population's mean-field inventory. To formally show this, recall that $dq_t^{j,\nu^j}=\nu_t^j\,dt$ for any $\nu^j\in\A^j$. Hence, from \eqref{eq:individual-optimal-control-real}, we have
\begin{equation}
	d\left( \qj{\nujst}_t - \qbar_t^k \right) = \tfrac{1}{2a_k} h_{2,t}^k \left( \qj{\nuj}_t - \qbar_t^k \right) \, dt \;.
\end{equation}
This admits the solution
\begin{equation}
	\left( \qj{\nujst}_t - \qbar_t^k \right) = \left( q_0^j - \qbar_0^k \right) e^{\frac{1}{2a_k} \int_0^t h_{2,u}^k du}
	\;.
\end{equation}
Since $h_{2,t}^k<0$, it is clear that $\lvert \qj{\nujst}_t - \qbar_t^k \rvert$ is a monotonically decreasing function of time.

\section{The \texorpdfstring{$\epsilon$}{Epsilon}-Nash Equilibrium Property}
\label{sec:epsilon-Nash}

In the previous section, we explicitly constructed the unique optimal trading actions of all agents in infinite population limit -- the optimal actions in the MFG. In this section, we explore the properties resulting from applying the MFG optimal controls to the finite player game. More specifically, we show that the controls satisfy the $\epsilon$-Nash equilibrium property.
\begin{definition}
	A set of controls $\left\{ \omega^j \in \A^j \,: j\in\mfN\right\}$ forms an $\epsilon$-Nash equilibrium with a collection of objective functionals $\left\{ J_j(\cdot,\cdot) \, : \, j\in \mfN\right\}$, if there exists $\epsilon>0$, s.t.
	\begin{equation}
		J_j(\omega^j , \omega^{-j}) \leq
		\sup_{\omega \in \A} J_j(\omega , \omega^{-j}) \leq
		J_j(\omega^j , \omega^{-j}) + \epsilon\,, \qquad \forall j\in\mfN.
	\end{equation}
\end{definition}

The definition of an $\epsilon$-Nash equilibrium characterizes a collection of controls that deviates no farther than $\epsilon$ from the Nash equilibrium of the collection of objective functions. We will prove that the optimal MFG controls obtained in section~\ref{sec:solving-the-mfg} satisfies the $\epsilon$-Nash property for any finite game with a large enough population size. In particular, we show that for any given $\epsilon>0$, there exists a population size $N_\epsilon$ so that the $\epsilon$-Nash property holds for any population of size $N>N_\epsilon$.

\begin{theorem}[$\epsilon$-Nash equilibrium] \label{thm:epsilon-nash-thm}
	Consider the collection of objective functions $\left\{ H_j \, : j \in \mfN \right\}$ defined in equation~\eqref{eq:sub-pop-objective-rep} and the set of optimal mean-field controls $\{ \nujst \}_{j=1}^\infty$ defined in equation~\eqref{eq:individual-optimal-control-real}. Suppose that there exists a sequence $\left\{ \delta_N \right\}_{N=1}^{\infty}$ such that $\delta_N \rightarrow 0$ and
	\begin{equation}
		\left\lvert \tfrac{N_k^{\N}}{N} -  p_k \right\rvert = o(\delta_N)
	\end{equation}
	for all $k\in\mfK$, then
	\begin{equation}
		H_j(\nujst , \nu^{\ast,-j}) \leq
		\sup_{\nu \in \A} H_j(\nu , \nu^{\ast,-j}) \leq
		H_j(\nujst , \nu^{\ast,-j}) + o( \tfrac{1}{N} ) + o(\delta_N)
	\end{equation}
	for each $j\in\mfN$.
\end{theorem}
\begin{proof}
	The proof is found in Appendix~\ref{sec:proof-epsilon-nash}.
\end{proof}

Theorem~\ref{thm:epsilon-nash-thm} implies that for any fixed value of $\epsilon$, we can identify the minimum population size, $N_\epsilon$, so that the $\epsilon$-Nash property holds for all $N>N_\epsilon$. More specifically, the theorem states that the size of the quantity $N_\epsilon$ will grow as $\epsilon\rightarrow 0$ at a super-linear rate that is a function of the sequence $\delta_N$. The special case where the rate of growth of $N_\epsilon$ is exactly linear occurs when $o(\delta_N) \leq o(\frac{1}{N})$.

From a more intuitive standpoint, Theorem~\ref{thm:epsilon-nash-thm} simply tells us that the mean-field optimal controls are always a `good enough' substitute for the optimal finite-game controls provided that the population size is large enough.

It is important to note that in the finite payer game, agents cannot use the empirical mean-field in their individual strategies, since it is not measurable with respect to each agent's visible filtration. They instead generate a `fictitious' mean-field inventory process $(\bqbar_t)_{t\in[0,T]}$ which trades at the rate of $(\bnubar_t)_{t\in[0,T]}$ according to \eqref{eq:mean-field-structure}. These fictitious mean-field inventories then feed into the individual agent's trading strategy $\nujst$ according to the result in Theorem \ref{thm:mean-field-is-true}.

\section{Numerical Experiments}
\label{sec:Numerical-Experiments}

In this section we study the behaviour of the mean-field optimal controls through simulations of the finite player game. We consider a model where the asset price process is a mean-reverting pure-jump process with a latent (unobservable) process driving the dynamics. Agents must filter the value of the latent process from observed paths of the asset price, and to use this filter to make predictions on the future (expected) value of the asset price process. We conclude by exploring some of the properties of the agent's trading decisions by interpreting simulation results.

We begin by defining the un-impacted asset price process $F=(F_t)_{t\in[0,T]}$. The process $F_t$ is defined as
\begin{equation}
	F_t = \alpha \left( L_t^+ - L_t^- \right)\,,
\end{equation}
where $L^\pm = (L_t^\pm)_{t\in[0,T]}$ are counting processes each with respective stochastic intensity processes $\gamma^\pm = ( \gamma_t^\pm )_{t\in[0,T]}$, and the constant $\alpha >0$ represents the tick size for the asset price. The process $F$ is defined so that it may only jump up or down by a single tick during any small instant in time.

As previously mentioned, we wish the asset price process to be mean reverting and to include some latent component in its dynamics. To achieve this, we define each intensity process so that
\begin{equation}
	\gamma_t^\pm = \sigma + \kappa \left( \Theta_t - F_t \right)_{\pm}
	\;,
\end{equation}
where $(\,\cdot\,)_\pm$ represents the positive or negative part of its argument and $\Theta=\left( \Theta_t \right)_{t\in[0,T]}$ is a latent process. This specification causes the un-impacted asset price $F$ to mean-revert to $\Theta$. Lastly, we define the process $\Theta$ to be an $M$-state continuous time Markov chain with generator matrix\footnote{The generator matrix $\bm{C}\in\mathbb{R}^{M\times M}$ of a $M$-state continuous time Markov chain $\Theta$ has non-diagonal entries $\bm{C}_{i,j}\geq 0$ if $i\neq j$ and diagonal entries $\bm{C}_{i,i}= -\sum_{j\neq i} \bm{C}_{i,j}$. $\bm{C}$ is defined so that $\P \left( \Theta_t = \theta_j {\lvert} \Theta_0 = \theta_i \right) = \left( e^{t \, \bm{C}} \right)_{i,j}$, where $\left( e^{t \, \bm{C}} \right)_{i,j}$ is element $(i,j)$ of the matrix exponential of $t \, \bm{C}$.} $\bm{C}$, taking values in the set $\left\{ \theta_i \right\}_{i=1}^M$. We also assumed that the initial value of the latent process, $\Theta_0$, has prior distribution $\bm{\pi}=\{ \pi_m \}_{m=1}^M$, where $\pi_m =\P(\Theta_0=\theta_m)$. The un-impacted asset price process can be viewed as a pure-jump analogue to an Ornstein-Uhlenbeck process. The parameter $\kappa$ controls the strength of the mean reversion of $F$ towards $\Theta$, while the parameter $\sigma$ controls the base level of noise in the paths of $F$. We point the reader to \cite[Section 6]{casgrain_jaimungal_2016} for further exposition of this model.

As before, we assume there is a total population of $N$ players divided into $K$ sub-populations all trading the same asset $S$. The asset price process $S$ is assumed to be given by
\begin{equation}
	S_t^{\nubarN} = F_t + \lambda \, \qbar^{(N)}_t
	\;,
\end{equation}
which can also be recast in the semi-martingale representation as in equation~\eqref{eq:S-price-dynamics}. Each agent chooses their trading strategy according to the mean-field optimal control derived in Section~\ref{sec:solving-the-mfg}. Each of the terms that need to be computed for this control can be obtained in closed form up to inverses and the computation of inverses, matrix exponentials and ordered exponentials. We set the initial inventory values of all agents to i.i.d. Gaussian random variables. More specifically, for an agent $j$ in sub-population $k$, we assume that
\begin{equation}
	\mfQ_0^k \sim \mathcal{N}(\invmean_0^k,\bar{\sigma}_0^k)
\end{equation}
for constants $\{\invmean_0^k\}_{k=1}^K$ and $\{\bar{\sigma}_0^k\}_{k=1}^K$. Each agent participating in the game must compute the values of $\{\nubar^k\}_{k=1}^{K}$ and $\{\qbar^k\}_{k=1}^{K}$ to be able to determine their own trading strategy. The agents achieve this by using the results of Propositions~\ref{prop:solution-g1-g2} and~\ref{prop:agents-optimal-control} to compute $\bnubar_t$ and to evolve the value of the `fictitious' mean-field inventory process $\bqbar_t$. To compute $\bg_{1,t}$, each agent must compute the conditional expected value $\E \left[ A_u \big{\lvert} \F_t \right]$, where $A_t = \kappa \left( F_t - \Theta_t\right)$. To achieve this, agents use the observed path of $S^{\nubarN}$ up until time $t$ to compute the posterior distribution of the current value of the latent process $\Theta_t$, which they use to compute the expected future return on the asset price. The expected value can be computed in closed form for this latent alpha model and its solution is presented and discussed in detail in \cite[Section 6]{casgrain_jaimungal_2016}.

We perform simulations of a game with $K=2$ distinct (unequal) sub-populations and a total of $N=30$ agents. We allow agents to trade in the finite time interval ending at $T=1$, representing the length of one whole trading day. The Table~\ref{tab:sim-player-params} below lays out the parameters  for the starting distribution of inventories and for the objective function for each sub-population.
\begin{table}[htbp]
  \centering
    \begin{tabular}{r|rrrrll}
    \multicolumn{1}{l|}{$k$} & \multicolumn{1}{l}{$\invmean_0^k$} & \multicolumn{1}{l}{$\bar{\sigma}_0^k$} & \multicolumn{1}{l}{$N_k$} & \multicolumn{1}{l}{$\Psi_k$} & $\phi_k$ & $a_k$ \bigstrut[b]\\
\cline{2-7}    1     & 100   & 50    & 20    & 100   & $10^{-2}$ & $10^{-4}$ \bigstrut[t]\\
    2     & 0     & 50    & 10    & 100   & $10^{-3}$ & $10^{-4}$ \\
    \end{tabular}%
  \caption{Simulation Parameters for the two sub-population of agents.}
  \label{tab:sim-player-params}%
\end{table}%

We set the first sub-population of agents to be long the asset on average at the start of the game, whereas the second starts off holding no inventory on average. We also make sub-population 2 more inclined to trade on alpha by making their $\phi_1$ parameter (which controls their risk appetite) considerably smaller than for sub-population 1. Both sub-populations have the same instantaneous transaction cost parameter $a_k$, and we set the size of sub-population 2 to be twice that of sub-population 1. Lastly, we set the $\Psi_k$ parameter to be very large to force agents to fully liquidate their inventories by $t=T$.

In the simulation, we assume that the latent process can take $M=2$ possible states and that $\Theta_t \in \{ 4.95 , 5.05 \}$. This causes the price to either mean-revert up or down depending on the state of $\Theta_t$. The remaining asset price process parameters are presented in Table~\ref{tbl:sim-asset-params} below.
\begin{table}[h!]
  \centering
  \begin{tabular}{lllll}
  $F_0 = \$ 5$, & $\alpha = \$ 0.01$ & $\bm{\pi} = \left(\begin{smallmatrix} 0.5 \\ 0.5 \end{smallmatrix}\right)$ &
  $\bm C= \left[ \begin{smallmatrix} -1 & 1 \\ 1 & -1  \end{smallmatrix} \right]$, \\
  $\kappa = 360$, & $\sigma = 120.24$, & $\lambda = \$ 10^{-3}$.
  \end{tabular}%
  \caption{The parameters used for the asset price dynamics and for the latent process.}
  \label{tbl:sim-asset-params}
\end{table}

The asset price is set to start at $\$5$ and can mean-revert to either $\$4.95$ or $\$5.05$ over the course of the trading period. We set the tick value , $\alpha$, in this model to be 1 cent so that the un-impacted asset price may only jump by increments of this size. The values of $\bm{\pi}$ are chosen specifically so that agents have no particular preference for the starting value of the latent process. Furthermore, because of the choice of the generator matrix $\bm{C}$, the agent expects that the latent process will switch states (on average) once over the course of the trading period $t\in[0,1]$. The values of $\kappa$ and $\sigma$ are chosen to be relatively close to re-scaled values obtained from calibrating the model to market prices in \cite{casgrain_jaimungal_2016}. Lastly, the permanent impact parameter $\lambda$ is chosen to be 10 times larger than the temporary impact parameter $a_k$ -- which is line with the empirical studies in \cite{cartea2016incorporating}.

 During the simulation, we set a fixed path for $\Theta_t$ to be able to observe the agents' filtering performance. More specifically, we let $\Theta_t$ start off at $\Theta_0 = 4.95$, and then make it jump to $\Theta_t = 5.5$ at time $t=0.5$ and let it remain there until the end of the trading period. We show an example simulated path for the asset price and latent process, the agent's posterior distribution for the value of $\Theta_t$, and the individual agents' inventory paths over the course of the trading period in Figure~\ref{fig: Simulation PureJump Results 1}.
 \begin{figure}[h!]
    \centering
    \begin{subfigure}[t]{0.32\textwidth}
        \centering
        \includegraphics[width=\textwidth]{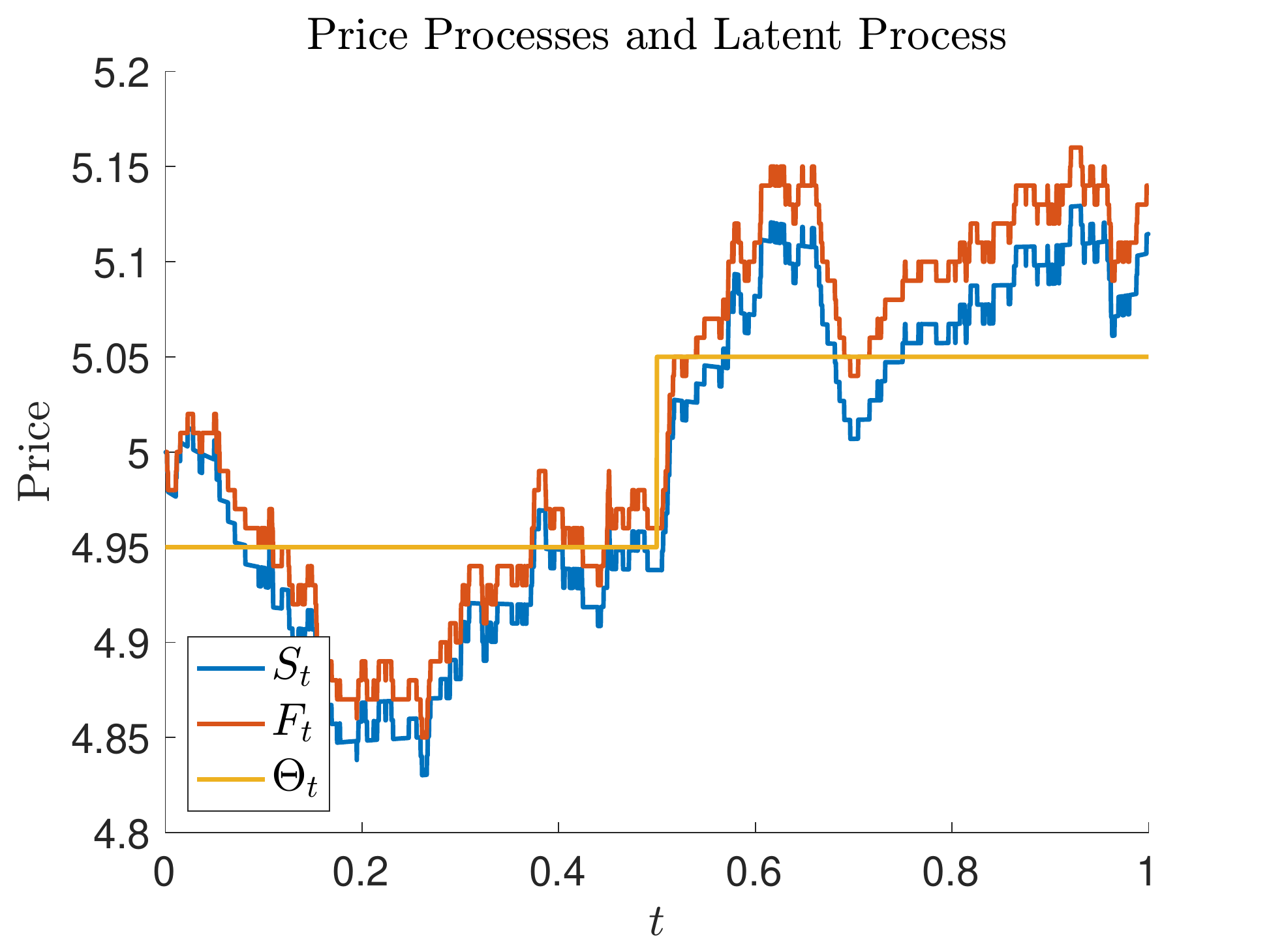}
    \end{subfigure}
    ~
    \begin{subfigure}[t]{0.32\textwidth}
        \centering
        \includegraphics[width=\textwidth]{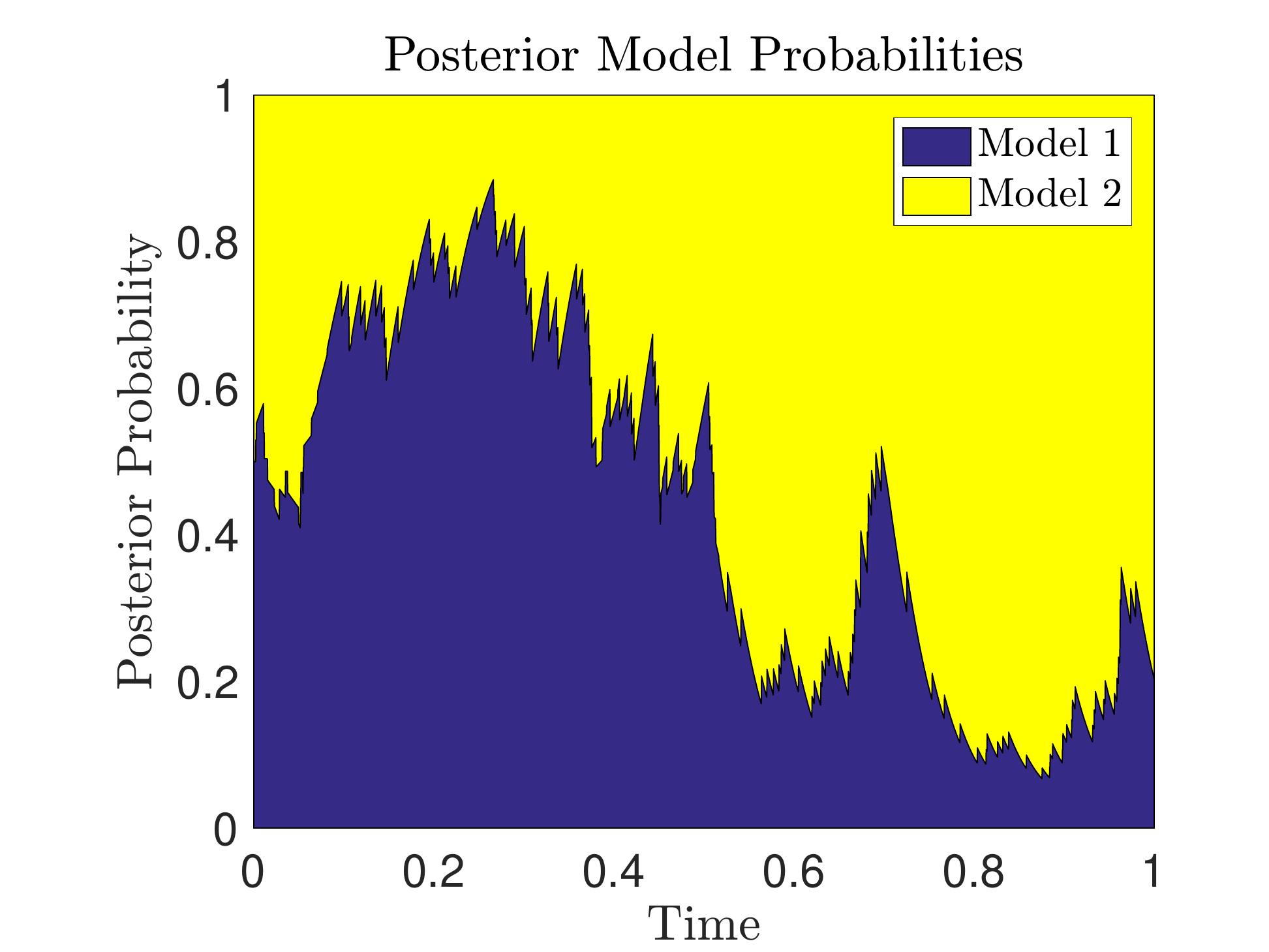}
    \end{subfigure}
    \begin{subfigure}[t]{0.32\textwidth}
        \centering
        \includegraphics[width=\textwidth]{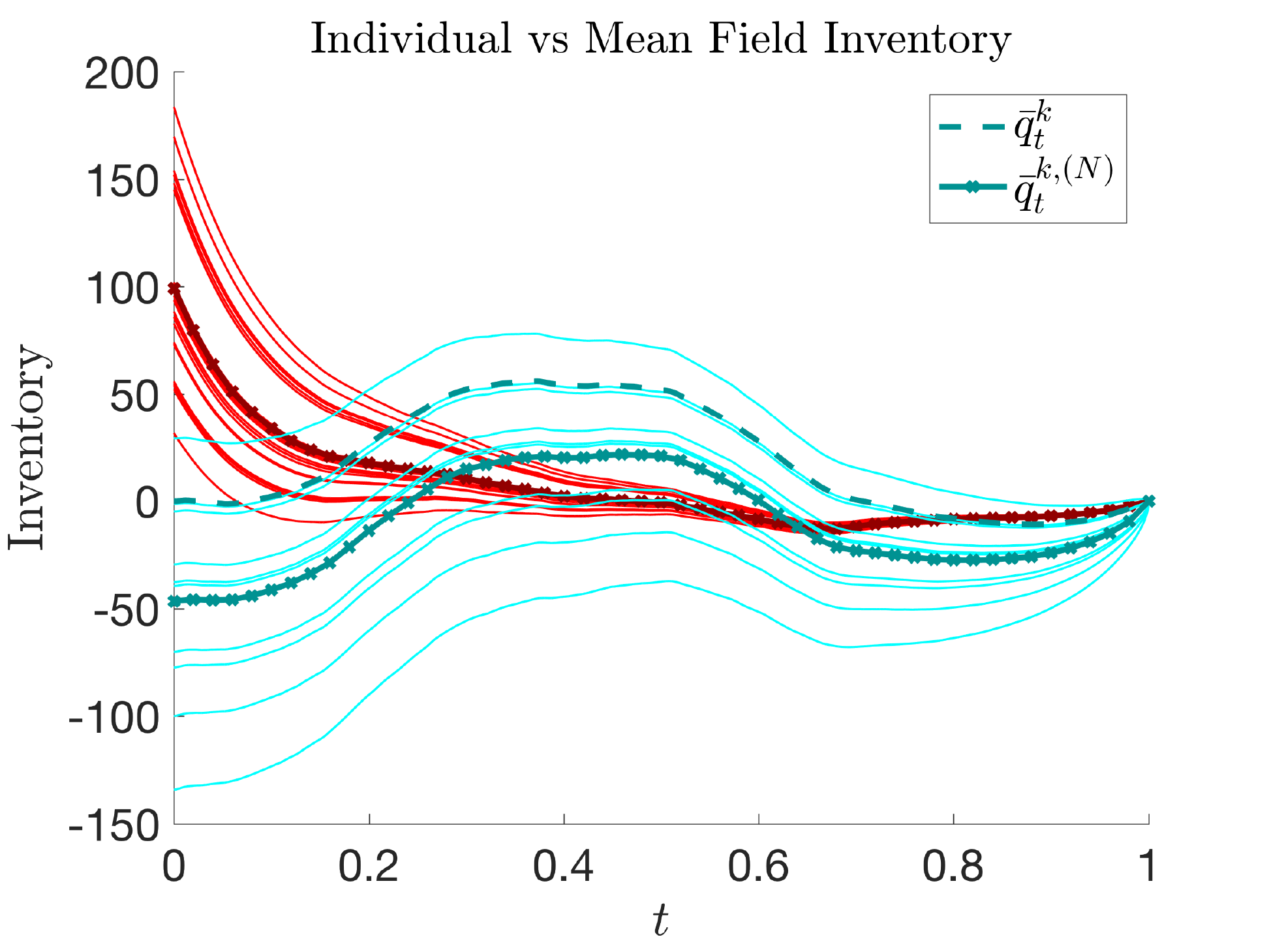}
    \end{subfigure}
    \caption{Simulated paths for the price process, latent process, the agent's filter, and the individual agent's inventory paths. In the right most panel: the red lines represents sub-population 1 and the blue lines represent sub-population 2. The thin lines represent the paths of an agent's inventory, the thick dashed lines represent the average of inventories for all agents within a sub-population. The thick crossed lines represent the mean-field process for a given sub-population.}
    \label{fig: Simulation PureJump Results 1}
\end{figure}

The left panel in Figure~\ref{fig: Simulation PureJump Results 1} shows how the asset price first mean-reverts downwards when $\Theta_t$ is in its lower state, and then mean reverts back up towards $\$5.05$ after $\Theta_t$ switches states at $t=0.5$. We also see from the centre panel of Figure~\ref{fig: Simulation PureJump Results 1} that agents are able to correctly learn the value of the latent process just by observing the paths of the asset price. The agents' posterior begins without a preference for the state, but then evolve to realize that the initial state of the latent process is at $\Theta_t = 4.95$. After $t=0.5$, the agent also identifies the switch in $\Theta_t$ and adjusts their posterior distribution accordingly. The price impact caused by the trading activity of agents can be seen in the left panel of Figure~\ref{fig: Simulation PureJump Results 1} as the difference between $F_t$ and $S_t$, which varies over time.


The right panel of Figure~\ref{fig: Simulation PureJump Results 1} displays the inventory values of agents over the trading period. Agents from sub-population $1$ (red) start with (on average) a higher inventory value than those in sub-population $2$ (blue). Agents in sub-population $1$ have inventories that very quickly converge towards the sub-population mean-field, which itself tends to zero quickly. Agents in sub-population $2$ also have inventories that mean-revert to their sub-population mean-field, but at a much slower rate than those in sub-population $1$. This behaviour is consistent with the observation in Section~\ref{sec:properties-of-the-optimal-control} that an individual agent's inventory and trading rate will always tend towards their sub-population's mean-field levels. Furthermore, the speed of the reversion towards the sub-population mean is controlled by the function $h_{2,t}^k$, which is defined in Proposition~\ref{prop:agents-optimal-control}. Larger values of the parameter $\phi_k$ increase the magnitude of $h_{2,t}^k$, and an increase in the magnitude of $h_{2,t}^k$ increases the speed of the reversion towards the mean-field as pointed out in Section~\ref{sec:properties-of-the-optimal-control}. Therefore, the difference in the speed of mean-reversion for the two sub-populations can be explained by differences in the value of the parameter $\phi_k$.


Moreover, agents in sub-population $2$ trade on the latent alpha considerably more than those in sub-population $1$. Members of sub-population $1$ seek to liquidate their inventories quickly, while members of sub-population $2$ are more inclined to take on inventory exposure due to deviation of the asset price from its filtered mean-reversion level. This also ties back to the value of the parameter $\phi_k$, which controls the agent's risk appetite. Since a lower value of $\phi_k$ corresponds to a higher risk tolerance, we see more alpha trading stemming from sub-population $2$.

The latent alpha trading observed in sub-population $2$ also matches the posterior probabilities (center panel of Figure~\ref{fig: Simulation PureJump Results 1}). When the agents estimate that the latent process is in the lower mean-reversion state, they begin taking on a long position in anticipation of a switch from the lower state to the upper state, which they expect will occur at least once before the end of the trading period. After the switch occurs, they begin reversing their position to a net short, expecting the reverse behaviour. The net short that agents take has a lower magnitude than the net long that agents previously took because less time remains until the end of the trading period, since it is both less likely that they will witness another switch and there is less time remaining before they are forced to completely liquidate their inventories. As the trading period nears its end, we see agents gradually reduce their exposure to zero so that they are flat by the end of the trading period.

\section{Conclusion}
\label{sec:Conclusion}

In this paper, we presented a stochastic game model for a market in which a finite population of players divided into heterogenous sub-populations trades a single asset. Agents have access to incomplete information of the market and of the actions of other agents, and we derive the mean-field game in the limit of infinite number of players. Using techniques from convex analysis allows us to obtain closed-form solutions for the mean-field of each sub-population and the optimal action for each individual agent. We then show that the solution obtained by solving the mean-field game in fact satisfies the $\epsilon$-Nash property in the finite player version of the game. Lastly, we present a simulated example of the finite stochastic game and analyze some of the agent's behaviour.

There are a number of future directions for this research. Here we outline a few directions, which is by no means exhaustive. One direction, which we have already begun investigating, is account for model-heterogeneity. That is, to allow agents in different sub-population to belief in different models, in addition to having heterogeneous preferences. Another direction of research is to restrict agents to trade at stopping times, rather than continuously as done here. Lastly, accounting for model uncertainty/ambiguity aversion along the lines of \cite{cartea2016incorporating} and \cite{huang2017mean} would be a very interesting direction to explore so that agents' strategies become more robust.

\newpage

\appendix

\footnotesize

\section{Proofs for Section~\ref{sec:solving-the-mfg} -- The Optimal Control Problem}

\subsection{Proof of Lemma~\ref{thm:predictable-representation}}
\label{sec:proof-thm-predictable-representation}

\begin{proof}
	First, let us define $F_t = S_t^{\nubar} - \int_0^t \lambda \nubar_u \, du$. Since $\nubar_t\in\HT$ and $\nubar_t$ is $\F$-adapted, it is clear that $F_t \in \HT$ and that $F_t$ is $\F$-adapted.

	Let $\Ahat_t = \E\left[{A_t} \mid \F_t \right]$. By it's definition, $\Ahat_t$ is $\F$-adapted. Furthermore, by Jensen's inequality,
	\begin{align}
		\E \int_0^T \E\left[{A_t} \mid \F_t \right]^2 \, dt
		&\leq
		\E \int_0^T \E\left[ A_t^2 \mid \F_t \right] \, dt
		\;.
	\end{align}
	since the integrand is non-negative, we may apply Fubini's theorem and the tower property,
	\begin{align}
		\E \left[ \E\left[ \int_0^T  A_t^2  \, dt \mid \F_t \right] \right]
		= \E \left[ \int_0^T  A_t^2 \, dt \right] < \infty
		\;,
	\end{align}
	therefore $\Ahat\in\HT$.

	Next, let us define
	\begin{equation}
		\Mhat_t = F_t - \int_0^T \Ahat_t \, dt
		\;.
	\end{equation}
	Since $\Ahat$ and $F$ are $\F$-adapted, $M$ is also $\F$-adapted. Furthermore, for any $t\in[0,T]$
	\begin{align}
		\E[\Mhat_t^2]
		&\leq
		4 \left( \E [F_t^2] + \E \left[\textstyle\int_0^T \Ahat_u^2 \, du \right]\right)
\leq
		32\, \E \left[\textstyle\int_0^T \left(A_u^2 + M_u^2\right) du \right]
		< \infty
		\;,
	\end{align}
	which demonstrates that $\Mhat\in\LT$. Using the dynamics of $F$, for $0\leq t \leq u \leq T$
	\begin{align}
		\E\left[ \Mhat_u - \Mhat_t \mid \F_t \right]
		&= \E \left[ \left.  ( F_u - F_t) - \textstyle\int_u^t \Ahat_s \, ds\, \right| \F_t \right] \\
		&=
		\E \left[ \left.\textstyle\int_u^t ( A_s -  \Ahat_s) \, ds + (M_u - M_t)\,\right|\F_t \right]
=
		\E \left[ \left.\textstyle\int_u^t ( A_s -  \Ahat_s) \, ds \,\right|\F_t\right]
		\;,
	\end{align}
	applying Fubini's theorem and the tower property
	\begin{align}
		\E \left[ \left.\textstyle\int_u^t ( A_s -  \Ahat_s) \, ds \,\right| \F_t \right]
		=
		\E \left[ \left. \int_u^t \E \left[\left. A_s -  \Ahat_s \,\right|\F_s \right] \, ds \,\right| \F_t \right]
		= 0
		\;.
	\end{align}
	Therefore,
	\begin{equation}
		\E\left[ \Mhat_u \mid \F_t \right] = \Mhat_t
		\,
	\end{equation}
	which shows that $\Mhat$ is a martingale.
	From the definitions of $\Mhat$ and $\Ahat$, it is easy to verify that \eqref{eq:SinTermsOfMhat} is satisfied.
\end{proof}

\subsection{Proof of Lemma~\ref{lemma:obj-is-concave}}
\label{sec:appendix-proof-obj-is-concave}

\begin{proof}
	To show that the functionals $\Hbar_j$ are strictly concave, we must show that for any $0<\rho<1$, and $\nu,\omega\in\A^j$ where $(\mathbb{P} \times \mu) ( \nu_t \neq \omega_t) > 0 $, that
	\begin{equation} \label{eq:proof-concave-obj-1}
		\Hbar_j(\rho \, \nu+(1-\rho) \, \omega) - \rho \, \Hbar_j(\nu) - (1-\rho) \, \Hbar_j(\omega) > 0
		\;.
	\end{equation}
First, observe that $q^{j,\cdot}$ is linear in controls:
	\begin{equation*}
		\qj{\rho\nu+(1-\rho)\omega}_t = \rho \, \qj{\nu}_t + (1-\rho) \qj{\omega}_t
		\;,
	\end{equation*}
	for all $0<\rho<1$ and $\nu,\omega\in\A^j$. If we let
	\begin{equation*}
		\bm\Gamma_k =
		\begin{pmatrix}
			a_k & \Psi_k \\ \Psi_k & \phi_k
		\end{pmatrix}
		\;,
	\end{equation*}
	and expand the left side of the inequality \eqref{eq:proof-concave-obj-1}, we may use the linearity of $q^j$ to cancel out constant terms and $\qj{\cdot}_t (\Ahat_t + \lambda \, \nubar_t)$ terms. This yields
	\begin{align*}
	\text{left part of \eqref{eq:proof-concave-obj-1}}
		=&\E\Bigg[ \int_0^T
		 \rho \smallqfm{\nu}^\T \bm\Gamma_k \smallqfm{\nu} \\
		&\quad+
		(1-\rho) \smallqfm{\omega}^\T \bm\Gamma_k \smallqfm{\omega} \\
		&\quad-
		\left( \rho \smallqfm{\nu} + (1-\rho) \smallqfm{\omega} \right)^\T
		\bm\Gamma_k
		\left( \rho \smallqfm{\nu} + (1-\rho) \smallqfm{\omega} \right)
		\, dt \Bigg]
\\
\text{(completing the square)}		=& \E\Bigg[ \int_0^T
		\rho\,(1-\rho)
		\left( \smallqfm{\nu} - \smallqfm{\omega} \right)^\T
		\bm\Gamma_k
		\left( \smallqfm{\nu} - \smallqfm{\omega} \right)
		\, dt \Bigg]
		\;.
	\end{align*}
	expanding the above, and letting $\Delta_t = \nu_t - \omega_t$ and since $q^{\Delta}_t=\qj{\nu}_t-\qj{\omega}_t$,
	\begin{equation} \label{eq:proof-convex-sum-int}
		=
		\rho\,(1-\rho)\,
		\E\left[ \int_0^T
		\left\{
		a_k \Delta_t^2 + \phi_k \left( q_t^{\Delta} \right)^2 + 2\Psi_k \Delta_t q_t^{\Delta}
		\right\} \,dt
		\right] \;.
	\end{equation}
	Since $0\leq\rho\leq1$, we only need to demonstrate that the inside of the expected value is greater than zero. Since $\phi_k \geq 0$, we can guarantee that the middle term in \eqref{eq:proof-convex-sum-int} is $\geq 0$. Next, we may look at the right-most term in equation~\eqref{eq:proof-convex-sum-int}. Since we can write $q_t^{\Delta} = \int_0^t \Delta_u \, du$, integrating by parts yields
	\begin{equation}
		\E\int_0^T 2 \Delta_t q_t^{\Delta} \, dt
		=
		\E\left[
		\left( q_T^\Delta \right)^2 \right]
		\geq 0
		\;.
	\end{equation}
	Since $\Psi_k \geq 0$, this last result implies that the right-most term in \eqref{eq:proof-convex-sum-int} is $\geq 0$. Lastly, notice that if $(\mathbb{P} \times \mu) ( \nu_t \neq \omega_t) > 0 $, then
	\begin{equation}
		\E\left[ \int_0^T
		\Delta_t^2
		\; dt \right] >0
		\;.
	\end{equation}
	Since $a_k>0$, this last comment shows that \eqref{eq:proof-convex-sum-int} is strictly greater than zero.

\end{proof}

\subsection{Proof of Lemma~\ref{lemma:GateauxDeriv}}
\label{sec:appendix-proof-lemma-differentiable}

\begin{proof}
	Using the definition of the G\^ateaux derivative,
	\begin{equation}
		\left\langle \D\Hbar_j(\nu) , \omega \right\rangle =
		\lim_{\epsilon\searrow0}
		\frac{\Hbar_j (\nu + \epsilon\,\omega) - \Hbar_j (\nu)}{\epsilon}
	\end{equation}
	we will show that this limit exists and is equal to the result provided in the lemma. Using the representation for the objective $\Hbar_j$~\eqref{eq:F-adapted-objective} and canceling out the $t=0$ terms and using the linearity of the process $\qj{\nu}_t - \qj{\nu}_0$ in the variable $\nu$, we have
\begin{align}
	\Hbar_j(\nu+\epsilon\,\omega) - \Hbar_j(\nu)
	&=
	\epsilon \,
	\E\left[
	\int_0^T
	\left\{
	( \qj{\omega}_t - \qj{\omega}_0 ) (\Ahat_t + \lambda \, \nubar_t)
	- 2
	\left(\begin{smallmatrix} \nu_t \\ \qj{\nu}_t \end{smallmatrix}\right)^\T
	\bm\Gamma_k	
	\left(\begin{smallmatrix} \omega_t \\ \qj{\omega}_t - \qj{\omega}_0 \end{smallmatrix}\right)
	\, dt
	\right\}
	\right]
	\\&-
	\epsilon^2
	\E\left[
	\int_0^T
	\left(\begin{smallmatrix} \omega_t \\ \qj{\omega}_t - \qj{\omega}_0 \end{smallmatrix}\right)^\T
	\bm\Gamma_k
	\left(\begin{smallmatrix} \omega_t \\ \qj{\omega}_t - \qj{\omega}_0 \end{smallmatrix}\right)
	\, dt
	\right]
	\;,
\end{align}
where
\begin{equation*}
	\bm\Gamma_k =
	\begin{pmatrix}
		a_k & \Psi_k \\ \Psi_k & \phi_k
	\end{pmatrix}
	\;.
\end{equation*}
Dividing by $\epsilon$ and taking the limit yields
\begin{equation} \label{eq:j-gateax-1}
	\left\langle \D \Hbar_j(\nu) , \omega \right\rangle
	=
	\E\left[
	\int_0^T
	\left\{
	(\qj{\omega}_t - \qj{\omega}_0 )( \Ahat_t + \lambda \nubar_t )
	- 2
	\left(\begin{smallmatrix} \nu_t \\ \qj{\nu}_t \end{smallmatrix}\right)^\T
	\bm\Gamma_k	
	\left(\begin{smallmatrix} \omega_t \\ \qj{\omega}_t - \qj{\omega}_0 \end{smallmatrix}\right)
	\right\} \, dt
	\right]
	\;.
\end{equation}
Expanding the right part of the integrand in \eqref{eq:j-gateax-1} and re-grouping terms,
\begin{equation} \label{eq:j-gateax-2}
	\left\langle \D \Hbar_j(\nu) , \omega \right\rangle
	=
	\E\left[
	\int_0^T
	( \qj{\omega}_t - \qj{\omega}_0 ) \left( \Ahat_t + \lambda \nubar_t - 2 ( \phi_k \qj{\nu}_t + \Psi_k \nu_t ) \right)
	 dt
	-2
	\int_0^T
	\omega_t \left(
	a_k \nu_t + \Psi_k \qj{\nu}_t
	\right)  dt
	\right].
\end{equation}

Since $\nu,\omega\in\A^j$ and $\nubar,\hat{A}\in\HT$, the sufficient conditions for Fubini's theorem are met. Applying Fubini's theorem and the tower property
\begin{align}
\left\langle \D \Hbar_j(\nu) , \omega \right\rangle	 &=
	\int_0^T
	\E\left[
	\omega_t \, \left(
	-2a_k \nu_t - 2 \Psi_k \qj{\nu}_T
	+ \int_t^T \left\{  \Ahat_t + \lambda \nubar_t - 2\phi_k \qj{\nu}_u \right\} \,du  \right)
	\right] \, dt
	\\ &=
	\int_0^T
	 \,
	\E\left[
	\omega_t \,
	\E\left[
	-2a_k \nu_t - 2 \Psi_k \qj{\nu}_T
	+ \int_t^T \left\{ \Ahat_t + \lambda \nubar_t - 2\phi_k \qj{\nu}_u  \right\}\,du
	\, \Bigg\lvert \, \F_t
	\right]
	\right] \, dt
	\\ &=
	\E\left[
	\int_0^T
	\omega_t \,
	\E\left[
	-2a_k \nu_t - 2 \Psi_k \qj{\nu}_T
	+ \int_t^T \left\{ \Ahat_t + \lambda \nubar_t - 2\phi_k \qj{\nu}_u  \right\} \,du
	\, \Bigg\lvert \, \F_t
	\right] \, dt
	\right]
	\;,
\end{align}
which gives the desired result.
\end{proof}

\subsection{Proof of Proposition~\ref{thm:optim-bsde-prop}}
\label{sec:proof-thm-optim-bsde-prop}

\begin{proof}
	By using lemmas~\ref{lemma:obj-is-concave} and~\ref{lemma:GateauxDeriv} we may apply the results of \cite[Section 5]{ekeland1999convex} which state that if
	\begin{equation}
		\langle \D \Hbar_j(\nujst) , \omega \rangle = 0
	\end{equation}
	for all $\omega\in\A^j$ if and only if
	\begin{equation}
		\nujst = \arg\sup_{\nujst\in\A^j} \Hbar_j(\nu)
		\;,
	\end{equation}
	the strict concavity of $\Hbar$ implies that $\nujst$ must be unique up to $\P\times\mu$ null sets. Therefore all we need to demonstrate is that the derivative vanishes if and only if it is the solution to the stated FBSDE.

	\textbf{Sufficiency:} Let us suppose that $\nujst$ is the solution to the FBSDE in the statement of the proposition and that $\nujst\in\HT$. We need to show that $\nujst\in\A^j$ and that it makes the G\^ateaux derivative vanish.

	First, let us note that we may represent the solution to the FBSDE implicitly as
	\begin{equation} \label{eq:proof-prop-optim-FBSDE-eq0}
		2 \, a_k \, \nujst_t = \E \left[ - 2 \,\Psi_k \,\qj{\nujst}_T
	+ \int_t^T \left\{ \Ahat_u + \lambda \,\nubar_u - 2\phi_k\, \qj{\nujst}_u  \right\} \,du \Big\lvert \F_t^j \right]
	\;,
	\end{equation}
	which demonstrates that $\nujst$ is $\F^j$-adapted. Therefore, since $\nujst\in\HT$ and $\nujst$ is $\F^j$-adapted, we have that $\nujst \in \A^j$.

	Lastly we show that $\nujst$ makes the G\^ateaux derivative vanish. By plugging \eqref{eq:proof-prop-optim-FBSDE-eq0} into the expression for the G\^ateaux derivative from lemma~\ref{lemma:GateauxDeriv} and using the tower property, we find that it vanishes almost surely.

	\textbf{Necessity:} Let us suppose that $\langle \D \Hbar_j(\nujst) , \omega \rangle = 0$ for all $\omega\in\A^j$. This implies that
	\begin{equation} \label{eq:proof-prop-optim-FBSDE-eq1}
	\E\left[
	-2a_k {\nujst}_t - 2 \Psi_k \qj{\nujst}_T
	+ \int_t^T \left\{ \Ahat_u + \lambda \nubar_u - 2\phi_k \qj{\nujst}_u  \right\} \,du \Big\lvert \F_t^j \right] = 0
	\end{equation}
	$\P\times\mu$ almost everywhere.

To see this, suppose that $\langle \D \Hbar_j(\nujst) , \omega \rangle = 0$ for all $\omega\in\A^j$, but \eqref{eq:proof-prop-optim-FBSDE-eq1} does not hold. Then, choose
	\begin{equation} \label{eq:proof-prop-optim-FBSDE-eq2}
		\omega_t = \E\left[
	-2a_k {\nujst}_t - 2 \Psi_k \qj{\nujst}_T
	+ \int_t^T \left\{ \Ahat_u + \lambda \nubar_u - 2\phi_k \qj{\nujst}_u  \right\} \,du \Big\lvert \F_t^j \right]
	\;.
	\end{equation}
	First, it is clear that this choice of $\omega$ is $\F_t^j$ adapted by its very definition. Second, using the fact that $\nubark \in \HT$, by using Jensen's inequality and the triangle inequality on \eqref{eq:proof-prop-optim-FBSDE-eq2}, we can obtain the bound
	\begin{align*}
		\E \int_0^T ( \omega_t )^2 \, dt
		&\leq
		4 \left( a_k^2 + T \, \Psi_k^2 + T^2 \phi_k^2  \right) \E \int_0^T ( \nujst_t )^2 \, dt
		\, +
		\E \int_0^T \left( \Ahat_t^2 + \lambda^2 \nubar^2_t \right) \, dt
		\\ &<\infty
		\;,
	\end{align*}
	which implies that $\omega\in \HT$ and therefore $\omega\in\A^j$. When we plug this choice of $\omega$ into the expression for the  G\^ateaux derivative, we see that $\langle \D\Hbar_j(\nujst) , \omega \rangle > 0$, which contradicts the assumption that $\langle \D \Hbar_j(\nujst) , \omega \rangle = 0$ for all $\omega\in\A^j$.

	Using \eqref{eq:proof-prop-optim-FBSDE-eq1} and noting that $\nujst_t\in\F_t^j$, we may write
	\begin{equation}
		2 \, a_k \, \nujst_t = \E \left[ - 2 \,\Psi_k \,\qj{\nujst}_T
	+ \int_t^T \left\{ \Ahat_u + \lambda \,\nubar_u - 2\,\phi_k \,\qj{\nujst}_u  \right\} \,du \,\Big\lvert \F_t^j \right]
	\;,
	\end{equation}
	and
	\begin{equation}
		2 \, a_k \, \mMbar^j_t = \E \left[ - 2 \,\Psi_k \,\qj{\nujst}_T
	+ \int_0^T \left\{ \Ahat_u + \lambda \,\nubar_u - 2\,\phi_k \,\qj{\nujst}_u  \right\} \,du \,\Big\lvert \F_t^j \right]
	\;,
	\end{equation}
	which solves the FBSDE in the statement of the proposition.
\end{proof}

\section{Proofs for Section~\ref{sec:solving-the-mfg} -- Solving the BSDEs}

\subsection{Proof of Proposition~\ref{prop:solution-g1-g2}}
\label{sec:proof-solution-g1-g2}

\begin{proof}
	The proof will is split in the following parts: We show that
	\begin{enumerate}[(a)]
		\item $\bg_{2,t}$ defined in the statement of the proposition is a bounded and is the unique solution the Riccati ODE~\eqref{eq:riccati-ODE}.
		\item $\bg_{1,t}$ defined in the statement of the proposition is the solution to the BSDE~\eqref{eq:linear-BSDE}.
		\item $\E\left[\int_0^T \, \bg_{1,t}^\T \bg_{1,t} \, dt \right] < \infty$.
	\end{enumerate}
	\textbf{Part (a).} Let us first point out that the ODE~\eqref{eq:riccati-ODE} is a matrix-valued non-symmetric Riccati-type ODE. We prove the claims concerning the ODE~\eqref{eq:riccati-ODE} by applying theorems and tools for non-symmetric Riccati ODEs found in the set of papers~\cite{freiling2000non} and \cite{freiling2002survey}. First of all, let us define $\bgt_{2,t}=\bg_{2,T-t}$.

	We will show that all of the claims   hold for $\bgt_{2,t}$, and hence also for $\bg_{2,t}$. From ODE~\eqref{eq:riccati-ODE}, we find that
	\begin{equation} \label{eq:proof-prop-riccati-ODE-backwards}
\left\{
		\begin{array}{rl}
		\partial_t \bgt_{2,t} \!\!\!&=
		 \left( \bm{\Lambda} + \bgt_{2,t} \right) \left( 2 \ba\right)^{-1} \bgt_{2,t}
		- 2\bm{\phi}
		\\
		\bgt_{2,0} \!\!\!&= -2\bm{\Psi}
		\end{array}{rl}
\right.
	\;.
	\end{equation}
	Our objective is now to apply \cite[Theorem~2.3]{freiling2000non} on $\tilde\bg_{2,t}$ to show the existence and boundedness of a solution. Using the notation of \cite{freiling2000non}, we define
	\begin{equation}
		B_{11} = \bm{0} ,\; B_{12} = -J ,\;
		B_{21} = - 2\bm{\phi} ,\; B_{22} = \bm{\Lambda} J
		\;.
	\end{equation}
	and $W_0 = -2\bm{\Psi}$, where $J = (2\ba)^{-1}$. To meet the requirements of theorem~2.3 in \cite{freiling2000non}, we must find $C,D \in \mathbb{R}^{K\times K}$, $C=C^\T$ so that $L + L^\T \leq 0$ and $C + D W_0 + W_0^\T D^\T > 0$, where
	\begin{equation}
		L =
		\begin{pmatrix}
			- 2 D\bm{\phi} & -C J + D \bm{\Lambda} J \\
			0 & -J^\T D
		\end{pmatrix}
		\;.
	\end{equation}

	Let $D = \bm{I}^{(K\times K)}$ and $C = 5\bm{\Psi}$. With these choices of $C,D$, and using the fact that $\Psi$ is a diagonal matrix with positive entries, we find that
	\begin{equation}
		C + D W_0 + W_0^\T D^\T = \bm{\Psi} > 0
		\;,
	\end{equation}
	which meets one of the necessary conditions. The choices of $C$ and $D$ also imply that the matrix $L$ takes the form
	\begin{equation}
		L =
		\begin{pmatrix}
			- 2 \bm{\phi} & - (5 \bm{\Psi} + \bm{\Lambda}) J \\
			0 & -J
		\end{pmatrix}
		\;.
	\end{equation}

	Now, let us note that $\det(L) = \det(-2\bm{\phi}) \times \det(-J)$. This directly implies that the set of eigenvalues of $L$ is the union of the set of eigenvalues of $-2\bm{\phi}$ and those of $-J$. Since $-2\bm{\phi} \leq 0$ and $-J < 0$, all of the eigenvalues of $L$ are guaranteed to be non-positive and at least one of them is guaranteed to be non-zero, which implies that $L<0$. Hence, $L + L^\T < 0$ which meets the second condition of \cite[Thm. 2.3]{freiling2000non}, and guarantees the existence of a solution to the ODE~\ref{eq:proof-prop-riccati-ODE-backwards} and hence of \eqref{eq:riccati-ODE}.

	Since the solution to $\bm{g}_{2,t}$ exists and is continuous on the interval $[0,T]$, it follows that it is also bounded on this interval. Furthermore, the existence and boundedness of the solution and \cite[Thm 3.1]{freiling2002survey} guarantees that the solution is also unique. Using the representation $\bgt_{2,t} = P_t \, Q_t^{-1}$ from \cite{freiling2002survey} and solving the appropriate linear ODE system for each, we obtain the solution presented in the statement of the theorem.

	\textbf{Part (b).} In this part we show that $\bg_1$ presented in the statement of the proposition solves the linear BSDE~\eqref{eq:linear-BSDE}. First let us consider the process $\bxi=(\bxi_t)_{t\in[0,T]}$ with $\bxi_t\in\R^{K\times K}$, defined as
	\begin{equation}
		\bxi_t = \bm{:} e^{\int_0^t \, \left( \bLambda + \bg_{2,s} \right) \left( 2 \ba \right)^{-1} \, da }\bm{:}
		\;,
	\end{equation}
	which is the unique solution to the matrix-valued ODE,
	\begin{equation} \label{eq:proof-solution-g1-g2-eq1}
		-d\bxi_t = \, -\bxi_{t} \left( \bLambda + \bg_{2,t} \right) \left( 2 \ba \right)^{-1}  \, dt,
	\end{equation}
	with the initial condition $\bxi_0= \bm{I}^{(K\times K)}$, where $\bm{I}^{(K\times K)}\in \R^{K\times K}$ is the identity matrix. Using the above ODE and the BSDE~\eqref{eq:linear-BSDE} to compute the dynamics of the process $\bxi_t \, \bg_{1,t}$, we find that
	\begin{equation}
		d(\bxi_t \, \bg_{1,t}) = \bxi_t \bm{1}^{(K\times 1)} \Ahat_t \, dt
		-\bxi_t \, d\bmMt_t
		\;,
	\end{equation}
	with the boundary condition $\bxi_T \, \bg_{1,T} = \bm{0}^{(K\times K)}$. We may solve the BSDE above explicitly to yield
	\begin{equation}
		\bxi_t \, \bg_{1,t} =
		\E \left[
		\int_t^T \, \bxi_u \bm{1}^{(K\times 1)} \Ahat_u \, du
		\, \Big\lvert \,  \F_t^j
		\right]
		\;.
	\end{equation}
	Since $\bxi_t$ is guaranteed to be positive definite, we multiply by $\bxi_t^{-1}$ on both sides to obtain the solution for $\bg_{1,t}$,
	\begin{equation}
		\bg_{1,t} =
		\E \left[
		\int_t^T \, \bxi_t^{-1} \bxi_u \, \bm{1}^{(K\times 1)} \Ahat_u \, du
		\, \Big\lvert \, \F_t^j
		\right]
		\;,
	\end{equation}
	where we may replace $\bxi_t^{-1} \bxi_u$ by the ordered exponential $\bm{:} e^{\int_t^u \, \left( \bLambda + \bg_{2,s} \right) \left( 2 \ba \right)^{-1} \, da }\bm{:}$
	to obtain the final solution.

	\textbf{Part (c).} Let $\lVert \cdot \rVert$ represent the euclidean norm in $\R^K$. Since $\bg_{2,t}$ is a bounded function, the time-ordered exponential $\bxi_t$ is positive definite and bounded over $[0,T]$. Therefore there exists a constant $c>0$ so that for any column vector $\bm{x}\in\R^{K}$
	\begin{equation}
	 	\sup_{t,u\in[0,T]} \lVert \bxi_t^{-1} \, \bxi_u \, \bm{x} \rVert^2 \leq c \, \lVert \bm{x} \rVert^2
	 	\;.
	\end{equation}
	Applying Jensen's inequality and Fubini's theorem, along with this last result, to solution for $\bg_{1,t}$, we find
	\begin{align*}
		\E \int_0^T \lVert \, \bg_{1,t} \, \rVert^2 \, dt
		&\leq
		\E \int_0^T \E \left[  \int_{t}^T \,
		\lVert
		\bxi_t^{-1} \bxi_u \, \bm{1}^{(K\times 1)} \Ahat_u
		\rVert^2 \, du \, \Big\lvert \, \F_t^j
		\right]\, dt
		\\ &\leq
		c K \,
		\E \int_0^T \E \left[ \int_{t}^T \,
		(\Ahat_u)^2 \, du \, \Big\lvert \, \F_t^j
		\right] \, dt
		\\ &\leq
		c K T \, \E
		\int_{0}^T \,
		(\Ahat_u)^2
		\, du
		< \infty
		\;,
	\end{align*}
	as desired.
\end{proof}


\subsection{Proof of Proposition~\ref{thm:ansatz-prop-summary}}
\label{sec:proof-ansatz-prop-summary}
\begin{proof}
	Plugging in the ansatz
	\begin{equation}
		\bnut_t = ( 2 \ba )^{-1} \left( \bg_{1,t} + \bg_{2,t} \, \bqt_t \right)
	\end{equation}
	into the FBSDE~\eqref{eq:stacked-FBSDE} yields the equation \eqref{eq:ansatz-dynamics}, which vanishes since $\bg_{1}$ and $\bg_{2}$ solve \eqref{eq:linear-BSDE} and \eqref{eq:riccati-ODE}, respectively. By the definitions of $\bg_{1}$ and $\bg_2$ the boundary condition is satisfied, and therefore $\bnut_t$ above solves the FBSDE~\eqref{eq:stacked-FBSDE}.

	Since the FBSDE~\eqref{eq:mean-field-ansatz-FBSDE} is just row $k$ of the vector FBSDE~\eqref{eq:stacked-FBSDE}, $\nutk$ is trivially the solution to \eqref{eq:mean-field-ansatz-FBSDE}.

	Lastly, we must show that $\nutk\in\bigcap_{j=1}^\infty \A^j$. Inspecting the definition of $\A^j$ and $\F_t^j$, we find that
	\begin{equation}
		\bigcap_{j=1}^\infty \A^j =\left\{ \nu\text{ is }\F_t\text{-predictable}\,,\; \nu \in \HT \right\}
		\;.
	\end{equation}
	Therefore, we must show that $\nutk\in\HT$ and that $\nutk$ is $\F$-predictable. First of all, since $\bg_{2,t}$ is deterministic and $\bg_{1,t}$ is $\F$-predictable, it is clear that $\nutk$ is $\F$-predictable. Next, notice that $\nutk\in\HT$ if
	\begin{equation}
		\E \int_0^T \bnut_t^\T \, \bnut_t \, dt < \infty
		\;.
	\end{equation}
	Using the fact that $d\bqt_t = \bnut_t \, dt$, we find that
	\begin{equation}
		d\bqt_t = ( 2 \ba )^{-1} \left( \bg_{1,t} + \bg_{2,t} \, \bqt_t \right) \, dt
		\;.
	\end{equation}
	Solving this SDE yields
	\begin{equation}
		\bqt_t = \bqt_0 \, \bxi_t + ( 2 \ba )^{-1} \int_0^t \, \bxi_u \, \bg_{1,u} \, du
	\end{equation}
	where
	\begin{equation}
		\bxi_t = \bm{I}^{K \times K} + \int_0^t \bxi_u \, ( 2 \ba )^{-1}  \bg_{2,u} \, du
		\;,
	\end{equation}
	for all $t>0$. Since $\ba$ is positive definite and $\bg_{2,t}$ is bounded, we find that $\bxi_t$ must also be continuous and bounded over $[0,T]$. Therefore, using the boundedness of $\bxi$, the triangle inequality and Jensen's inequality, there exists a constant $C>0$ such that
	\begin{align*}
		\bqt_t^\T \bqt_t
		&\leq C \left( \bqt_0^\T  \, \bqt_0  + \int_0^t \bg_{1,u}^\T \bg_{1,u} \, du \right)
		\\ &\leq
		 C \left( \bqt_0^\T  \, \bqt_0  + \int_0^T \bg_{1,u}^\T \bg_{1,u} \, du \right)
		\;.
	\end{align*}
	Now, integrating and taking the expected value,
	\begin{align*}
		\E \int_0^T \bqt_t^\T \bqt_t \, dt
		& \leq C T \left( \bqt_0^\T  \, \bqt_0  + \int_0^T \bg_{1,u}^\T \bg_{1,u} \, du \right)\;.
	\end{align*}
	Noting that $\bqt_0$ is bounded and that $\E\int_0^T \bg_{1,t}^\T \, \bg_{1,t} \, dt
	, \infty$, we find that
	\begin{equation}
		\E \int_0^T \bqt_t^\T \bqt_t \, dt < \infty
		\;.
	\end{equation}

	Now using this result and applying the triangle inequality and Jensen's inequality to the expression for $\bnut$,
	\begin{align*}
		\E \int_0^T \bnut_t^\T \, \bnut_t \, dt
		&\leq (2\ba)^{-2}
		\left(
		\E \int_0^T \bg_{1,t}^\T \bg_{1,t} \, dt
		+
		\E \int_0^T (\bg_{2,t} \bqt_t)^\T ( \bg_{2,t} \bqt_t ) \, dt
		\right)
		\\ &\leq
		\left(
		\E \int_0^T \bg_{1,t}^\T \bg_{1,t} \, dt
		+
		C \, \E \int_0^T \bqt_t^\T \bqt_t \, dt
		\right)
		< \infty
		\;,
	\end{align*}
	where we use the boundedness of $\bg_{2,t}$ in the second line to obtain $C>0$, and thus obtaining the desired result.
\end{proof}

\subsection{Proof of Proposition~\ref{prop:agents-optimal-control}}
\label{sec:proof-agents-optimal-control}

\begin{proof}
	To prove the claims made in the statement of the proposition, we must show that the stated form of $\nu_t^j$ solves the FBSDE~\eqref{eq:agent-FBSDE-limit-mfansatz}. First, by plugging in the ansatz
\begin{equation}
	\nuj_t = \nutk_t + \frac{1}{2 a_k} h_{2,t}^k \left( \qj{\nuj}_t - \qtk{\nutk}_t \right)
\end{equation}
into the FBSDE, we obtain the simplification
\begin{align*}
	-2 a_k \, d\nutk_t - d h_{2,t}^k \left( \qj{\nuj}_t - \qtk{\nutk}_t \right)
	- \frac{1}{4 a_k^2} \left( h_{2,t}^k \right)^2 \left( \qj{\nuj}_t - \qtk{\nutk}_t \right) \, dt
	=
	\left( \Ahat_t + \blambda^\T \, \bnut_t  - 2\phi_k \qj{\nuj}_t \right) \,dt
	- d\mathcal{M}_t^j
	\;.
\end{align*}
Plugging in the FBSDE for $\nutk_t$ from equation~\eqref{eq:mean-field-ansatz-FBSDE} and choosing $\mathcal{M}_t^j = \mMbar^k_t$, we can cancel out terms and obtain the equation
\begin{equation}
	0 =
	\left( \qj{\nuj}_t - \qtk{\nutk}_t \right)
	\left\{ d h_{2,t}^k + \left( \frac{1}{4 a_k^2} \left( h_{2,t}^k \right)^2 - 2 \phi \right) \, dt \right\}
	\;,
\end{equation}
which must hold almost surely for all values of $(\qj{\nuj}_t - \qtk{\nutk}_t)$. Therefore, solving for $h_{2,t}^k$ which will make the terms inside of the curly brackets vanish will also solve  FBSDE~\eqref{eq:agent-FBSDE-limit-mfansatz}. Therefore, setting the terms inside of the curly brackets to zero and inserting the appropriate boundary condition, we get the ODE
\begin{equation}
	\begin{cases}
		&-d h_{2,t}^k = \left( \frac{1}{4 a_k^2} \left( h_{2,t}^k \right)^2 - 2 \phi \right) \, dt
		\\
		&\hspace{1.2em}h_{2,T}^k = -2 \Psi^k
	\end{cases}
	\;.
\end{equation}
This last ODE is of the well studied Riccati-type with the solution presented in the statement of the theorem.

Next, we wish to demonstrate that $h_{2,t}^k \geq 0$. First, let us notice that since $t<T$ and $\gamma_k\geq 0$ that $\sinh(-\gamma_k (T-t)) \leq 0$ and $\cosh(-\gamma_k (T-t)) \geq 1$. Since $\xi_k, \Psi_k \geq 0$ we then get that
\begin{equation}
	\frac{
		\Psi_k \cosh\left( -\gamma_k (T-t)\ \right)
		- \xi_k \sinh\left( -\gamma_k (T-t)\ \right) }{
		 \xi_k \cosh\left( -\gamma_k (T-t)\ \right)
		- \Psi_k \sinh\left( -\gamma_k (T-t)\ \right)
		}
		\geq 0
		\;,
\end{equation}
and the desired result follows.

Lastly, we wish to show that $\nuj\in\A^j$. First, note that $\nuj$ it is sufficient to show that $\nuj - \nutk \in \A^j$, since $\nutk_t\in\A_j$. First, let $\Delta_t = \qj{\nuj}_t - \qtk{\nutk}_t$. From the statement of the proposition we get that
	\begin{equation}
		d\Delta_t = \frac{h_{2,t}^k}{2 \, a_k} \Delta_t \, dt
		\;,
\end{equation}
with the boundary condition $\Delta_0 = \mfQ_0^j - \invmean_0^k$. Since the $\frac{h_{2,t}^k}{2 \, a_k}$ is deterministic and the boundary condition is $\F^j$-adapted, it is clear that $\Delta_t$ is $\F^j$-adapted. We may solve the SDE directly to yield the solution
\begin{equation}
	\Delta_t = \left( \mfQ_0^j - \invmean_0^k \right) e^{\int_0^t \frac{h_{2,u}^k}{2 \, a_k} \, du}
	\;.
\end{equation}
Since $\mfQ_0^j$ has a bounded variance and $h_{2,t}$ is a bounded function, it is clear that $\Delta\in\HT$. Hence , $\Delta\in\A^j$. Now because $\nuj_t - \nutk_t = \frac{h_{2,t}^k}{2 \, a_k} \Delta_t$, and $h_{2,t}$ is a bounded and deterministic function, we find that $\nuj - \nutk \in \A^j$.
\end{proof}

\subsection{Proof of Theorem~\ref{thm:mean-field-is-true}}
\label{sec:proof-mean-field-is-true}

We begin by introducing the following lemma, which will be used in the proof of Theorem~\ref{thm:mean-field-is-true}.

\begin{lemma} \label{lemma:mean-field-ansatz-diff}
	 Let $\nu_t^j$ be the ansatz mean-field optimal control defined in Proposition~\ref{thm:ansatz-prop-summary} for an agent $j$ in sub-population $k$. Then
	 \begin{equation}
	 	\qj{\nuj}_t - \qtk{\nutk}_t = \left( \mfQ_0^j - \invmean_0^k \right) e^{\frac{1}{2 a_k} \int_0^t h_{2,u}^k \, du}
	 	\;,
	 \end{equation}
	 where $\mfQ_0^k$ is the initial value of $j$'s inventory, $\invmean_0^k = \E \mfQ_0^j$, and $h_{2,t}^k$ is the function defined in proposition~\ref{prop:agents-optimal-control} satisfying the property $h_{2,t} < 0$.
\end{lemma}
\begin{proof}
From Proposition~\ref{prop:agents-optimal-control}, we have
\begin{equation}
	\nu_t^j = \nutk_t + f_t^k \left( \qj{\nuj}_t - \qtk{\nutk}_t \right)
	\;,
\end{equation}
where we let $f^k_t = \frac{h_{2,t}^k}{2 a_k}$. Using the above equation and noting that $\partial_t \left( \qj{\nuj}_t - \qtk{\nutk}_t \right) = \nu_t^j - \nutk_t$, we get that
\begin{equation}
	\partial_t \left( \qj{\nuj}_t - \qtk{\nutk}_t \right) = f_t^k \left( \qj{\nuj}_t - \qtk{\nutk}_t \right)
	\;.
\end{equation}
Solving the above ODE with the initial condition $q_0^k = \mfQ_0^k$ and $\qt_0^k = \invmean_0^k$ yields the desired result.
\end{proof}

Now we proceed with the proof of Theorem~\ref{thm:mean-field-is-true}.

\begin{proof}
	To prove the first result of the theorem, we study the difference $\nubark_t - \nutk_t$, where
	\begin{equation}
		\nubark_t = \lim_{N\rightarrow\infty}\frac{1}{N_k^{\N}} \sum_{j\in\mfK_k^{\N}}
		\nuj_t
		\;.
	\end{equation}
	Using the ansatz for $\nuj_t$ from Proposition~\ref{prop:agents-optimal-control}, we get
	\begin{equation}
		\nuj_t - \nubark_t = \frac{h_{2,t}^k}{2 a_k} \left( \qj{\nuj}_t - \qtk{\nutk}_t \right)
		\;.
	\end{equation}
	Using the result from Lemma~\ref{lemma:mean-field-ansatz-diff}, this becomes
	\begin{equation}
		\nuj_t - \nubark_t =  \left( \mfQ_0^j - \invmean_0^k \right) f_t^k
	\;,
	\end{equation}
	where $f_t^k = \frac{h_{2,t}^k}{2 a_k} \, e^{\frac{1}{2 a_k} \int_0^t h_{2,u}^k \, du}$ is a bounded, continuous function. Taking the average over all $j\in\mfK_{k}^{N}$ and taking the limit, we see that
	\begin{equation} \label{eq:proof-mean-field-is-true-eq1}
		\lim_{N\rightarrow\infty}
		\left( \frac{1}{N_k^{\N}} \sum_{j\in\mfK_k^{\N}}
		\nuj_t \right) - \nutk_t
		=  f_t^k \, \lim_{N\rightarrow\infty} \left( \frac{1}{N_k^{\N}} \sum_{j\in\mfK_k^{\N}}
		\left( \mfQ_0^j - \invmean_0^k \right)
		\right)
		\;.
	\end{equation}
	Since the collection $\{ \mfQ_0^j \}_{j\in\mfK_k^{\N}}$ is a collection of independent random variables with $\E \mfQ_0^j = \invmean_0^k$ and bounded variance, we may apply the law of large numbers which makes the right limit in \eqref{eq:proof-mean-field-is-true-eq1} vanish almost surely and in $L^2$. Therefore computing the left limit we have that
	\begin{equation}
		\nubark_t - \nutk_t = 0
		\;,
	\end{equation}
	almost surely for all $t\in[0,T]$. Which implies that $\bnut_t = \bnubar_t$ almost surely for all $t\in[0,T]$.

	Since $\bnut_t = \bnubar_t$, we find that $\nut_t = \nubar_t = \lim_{N\rightarrow\infty} \frac{1}{N} \sum_{j=1}^N \nujst_t$. Since the proposed form of $\nujst_t$ also solves the FBSDE~\eqref{eq:mean-field-ansatz-FBSDE}, $\nut_t = \nubar_t$ almost surely and $\nujst\in\A^j$, we find that $\nujst_t$ also solves the optimality FBSDE from Theorem~\eqref{thm:optim-bsde-prop}. Hence, applying Theorem~\eqref{thm:optim-bsde-prop}, the collection $\{ \nujst \}_{j=1}^\infty$ is optimal and satisfies
	\begin{equation}
		\nujst = \arg\sup_{\nu\in\A^j} \Hbar_j (\nu)
		\;,
	\end{equation}
	for all $j$.
\end{proof}

\section{Proofs for Section~\ref{sec:epsilon-Nash} -- \texorpdfstring{$\epsilon$-Nash Property}{E-Nash Property}}

\subsection{Proof of Theorem~\ref{thm:epsilon-nash-thm}} \label{sec:proof-epsilon-nash}

We begin the proof of Theorem~\ref{thm:epsilon-nash-thm} by introducing two lemmas. The first is a lemma provides a closed-form expression for the difference of an agent's mean-field optimal control and it's own sub-population's mean-field inventory. The second is a lemma regarding the distance between the mean-field game objective $\Hbar_j$ and the finite player game objective $H_j$.

\begin{lemma} \label{lemma:mean-field-diff}
	 Let $\nujst$ be the mean-field optimal control for an agent $j$ in sub-population $k$. Then
	 \begin{equation}
	 	\qj{\nujst}_t - \qbar_t^k = \left( \mfQ_0^j - \invmean_0^k \right) e^{\frac{1}{2 a_k} \int_0^t h_{2,u} \, du}
	 	\;,
	 \end{equation}
	 where $\mfQ_0^k$ is the initial value of $j$'s inventory, $\invmean_0^k = \E \mfQ_0^j$, and $h_{2,t}$ is the function defined in proposition~\ref{prop:agents-optimal-control} satisfying $h_{2,t} < 0$.
\end{lemma}
\begin{proof}
The result is found by using Lemma~\ref{lemma:mean-field-ansatz-diff} along with Theorem~\ref{thm:mean-field-is-true}.
\end{proof}

\vspace{1cm}
	
\begin{lemma} \label{lemma:obj-diff-bound}
	Let $\nu \in \A^j$ be some arbitrary admissible control and $\nu^{\ast,-j}\in\A^{-j}$ be the collection $\nu^{-j,\ast}:=\left( \nu^{1,\ast} , \dots , \nu^{j-1,\ast},\nu^{j+1,\ast}, \dots, \nu^{N,\ast} \right)$ of mean-field optimal controls for all agents except for $j$. Then
	\begin{equation}
		\left\lvert H_j(\nu,\nu^{-j,\ast}) - \Hbar_j(\nu) \right\rvert
		= o(\delta_N) + o(\frac{1}{N})
		\;.
	\end{equation}
\end{lemma}	
\begin{proof}
	We will show that the claim holds by instead demonstrating the equivalent claim that
	\begin{equation} \label{eq:lemma-obj-diff-claim}
	 	\left\lvert H_j(\nu,\nu_{-j,\ast}) - \Hbar_j(\nu) \right\rvert^2
		= o(\delta_N^2) + o(\tfrac{1}{N^2})\;.
	\end{equation}
	 Using the representation for $\Hbar_j$~\eqref{eq:F-adapted-objective} and the representation for $H_j$~\eqref{eq:alternative-objective-def}, we find that the square of their difference is equal to
	\begin{equation} \label{eq:lemma-obj-diff}
		\left\lvert H_j(\nu,\nu^{-j,\ast}) - \Hbar_j(\nu) \right\rvert^2
		=
		\lambda^2 \, \E \left[ \int_0^T \, ( \nubarN_t - \nubar_t ) \, dt \right]^2
		\;.
	\end{equation}
	Therefore it is sufficient for us to show that quantity on the right side of the equation is $o(N^{-2}) + o(\delta_N^2)$. If we consider the expected value appearing in equation~\eqref{eq:lemma-obj-diff}, we can apply the definition of $\nubarN$ to decompose it as
	\begin{equation}
		\E \left[ \int_0^T \, ( \nubarN_t - \nubar_t ) \, dt \right]
		=
		\E \left[ \int_0^T \,  \frac{1}{N} \, \left( \nu_t - \nujst \right) +
		\frac{1}{N} \sum_{i=1}^N \left( \nu^{i,\ast} - \nubar_t \right) \, dt \right]
		\;,
	\end{equation}
	where $\nujst_t$ is the mean-field optimal control for agent $j$.
	Using the triangle inequality and Jensen's inequality on the left, we find that
	\begin{equation}
		\text{\eqref{eq:lemma-obj-diff}}
		\leq
		\frac{\lambda}{N^2} \, \E \left[ \int_0^T \left( \nu_t - \nujst_t \right)^2 \, dt \right]
		+
		\lambda \, \E \left[ \int_0^T \,  \frac{1}{N} \sum_{i=1}^N \left( \nu^{\ast,i} - \nubar_t \right) \, dt \right]^2
			\;.
	\end{equation}
 	$\nu, \nuj \in\A^j \subset \HT$ implies that $\E \left[ \int_0^T \left( \nu_t - \nujst_t \right)^2 \, dt \right] < \infty$, and so
	\begin{equation}
		\frac{\lambda}{N^2} \, \E \left[ \int_0^T \left( \nu_t - \nujst_t \right)^2 \, dt \right]
		= o(N^{-2})
		\;.
	\end{equation}
	
	At this point, all that remains is to investigate the term
	\begin{equation} \label{eq:lemma-squared-diff-1}
		\E \left[ \int_0^T \,  \frac{1}{N} \sum_{i=1}^N \left( \nu_t^{i,\ast} - \nubar_t \right) \, dt \right]^2
		\;.
	\end{equation}
	Using the notation $p_k^{\N} = \frac{N_k^{\N}}{N}$ and $\nubarkN_t = \frac{1}{N_k^{\N}} \sum_{i\in\mfK_k^{\N}} \nu_t^{i,\ast}$, we may write
	\begin{align*}
		  \frac{1}{N} \sum_{i=1}^N \left( \nu_t^{i,\ast} - \nubar_t \right)
		  &=
		  \sum_{k=1}^K \left\{ p_k^{\N} \nubarkN_t  - p_k \nubark_t \right\}
		  \\ &=
		  \sum_{k=1}^K \left\{ \nubarkN_t \left( p_k^{\N} - p_k  \right) + p_k \left( \nubarkN_t - \nubark_t \right)  \right\}
		  \;.
	\end{align*}
	Using this last result and the triangle inequality, we get
	\begin{equation} \label{eq:lemma-squared-diff-2}
		\text{lhs of \eqref{eq:lemma-squared-diff-1}}
		\leq \sum_{k=1}^K
		\left\{
		\E \left[ \int_0^T \nubarkN_t \, dt \right]^2 \, \left( p_k^{\N} - p_k \right)^2
		+
		p_k \, \E\left[ \int_0^T \left( \nubarkN_t - \nubark_t \right) \, dt \right]^2
		\right\}
	\;.
	\end{equation}
	First, by plugging in the result Lemma~\ref{lemma:mean-field-diff} taking the average over all $i\in\mfK_k^{\N}$ to compute $\nubarkN - \nubark$ we get
	\begin{align}
		\E\left[ \int_0^T \left( \nubarkN_t - \nubark_t \right) \, dt \right]
		&=
		\left( \int_0^T \, e^{\int_0^t \tfrac{h_{2,s}}{2 a_k} \, ds} \, dt \right) \,  \frac{1}{N_k^{\N}} \sum_{i\in \mfK_k^{\N}} \E\left[ \mfQ_0^j - \invmean_0^k \right] \,
		\\ &=
		0
		\;,
	\end{align}
	which implies that $\E \left[ \int_0^T \nubarkN_t \, dt \right] = \E \left[ \int_0^T \nubar_t^k \, dt \right]$. Applying this to \eqref{eq:lemma-squared-diff-2}, and noting that $\nubark_t\in \HT$ we get
	\begin{equation}
		\text{\eqref{eq:lemma-squared-diff-2}}
		\leq
		C_0 \sum_{k=1}^K \left( p_k^{\N} - p_k \right)^2 = o(\delta_N^2)
	\end{equation}
	for some $C>0$. Putting this all back together, we find that
	\begin{equation}
		\left\lvert H_j(\nu,\nu_{-j}) - \Hbar_j(\nu) \right\rvert^2
		=
		o(\delta_N^2)
		+ o(N^{-2})
		\;,
	\end{equation}
	for some other constant $C_1>0$. Taking the root of both sides, and noting that
	\begin{equation}
		 \sqrt{o(\delta_N^2) + o(N^{-2})}
		 =
		 o(\delta_N) + o(N^{-1})
	\end{equation}
	 we obtain the final result.
\end{proof}

\subsection{Proof of Theorem~\ref{thm:epsilon-nash-thm}}

\begin{proof}
	We prove the result of the theorem by using the Lemma~\ref{lemma:mean-field-diff}. First, let us note that by the definition of the supremum,
	\begin{equation}
		H_j(\omega,\nu^{-j,\ast}) \leq \sup_{\nu\in\A^j} H_j(\nu,\nu^{-j,\ast})
	\end{equation}
	holds for all $\omega\in\A^{j}$, and therefore the left-most inequality in the statement of Theorem~\ref{thm:epsilon-nash-thm} holds.

	Next, we must show that the right-most inequality in the statement of Theorem~\ref{thm:epsilon-nash-thm} also holds. First let us note that by Lemma~\ref{lemma:mean-field-diff}, for any $\nu\in\A^j$,
	\begin{align}
		H_j(\nu,\nu^{-j,\ast})
		&\leq
		\Hbar_j(\nu) + o(\delta_N) + o(N^{-1})
		\\ &\leq
		\Hbar_j(\nujst) + o(\delta_N) + o(N^{-1})
		\;,
	\end{align}
	 where we use the fact that $\Hbar_j(\nujst) = \sup_{\nu\in\A^j} \Hbar_j(\nu)$. Applying Lemma~\ref{lemma:mean-field-diff} again, we find that
	\begin{equation}
			H_j(\nu,\nu^{-j,\ast}) \leq H_j(\nuj,\nu^{-j,\ast}) + 2 \, o(\delta_N) + 2 \, o(N^{-1})
			\;.
	\end{equation}
	Since the above inequality holds for all $\nu\in\A^j$ we may take the supremum on the left, and cancel out the constant terms multiplying the little-$o$ terms to yield the final result,
	\begin{equation}
		\sup_{\nu\in\A^j} H_j(\nu,\nu^{-j,\ast}) \leq H_j(\nuj,\nu^{-j,\ast}) + o(\delta_N) + o(N^{-1})
		\;.
	\end{equation}
\end{proof}

\clearpage
\bibliographystyle{siamplain}
\bibliography{StochasticGamesPartialBib}

\nocite{*}

\end{document}